\pdfoutput=1
\documentclass[sigconf]{acmart}
\newcommand{\ABHINAV}[1]{}
\renewcommand\footnotetextcopyrightpermission[1]{} 
\usepackage{bbding}
\usepackage{times}
\usepackage{adjustbox}
\usepackage{graphics}
\usepackage{graphicx}
\usepackage{subcaption} 
\usepackage{natbib}
\usepackage{algorithm}
\usepackage{algorithmic}
\usepackage{todonotes}
\usepackage{amsmath}
\usepackage{amsfonts}
\usepackage{booktabs}
\usepackage{amssymb}
\usepackage{pifont}

\newcommand{\RR}{\mathbb{R}}

\setcopyright{none}
\usepackage{graphicx}
\usepackage{balance}

\usepackage{cleveref}
\crefname{section}{§}{§§}
\Crefname{section}{§}{§§}
\title[GossipGraD]
{GossipGraD: Scalable Deep Learning using Gossip Communication based Asynchronous Gradient Descent}

\author{Jeff Daily}
\affiliation{%
    \institution{Pacific Northwest National Laboratory}
}
\email{jeff.daily@pnnl.gov}
\author{Abhinav Vishnu}
\affiliation{%
    \institution{Pacific Northwest National Laboratory}
}
\email{abhinav.vishnu@gmail.com}
\author{Charles Siegel}
\affiliation{%
    \institution{Pacific Northwest National Laboratory}
}
\email{charles.m.siegel@gmail.com}
\author{Thomas Warfel}
\affiliation{%
    \institution{Pacific Northwest National Laboratory}
}
\email{thomas.warfel@pnnl.gov}
\author{Vinay Amatya}
\affiliation{%
    \institution{Pacific Northwest National Laboratory}
}
\email{vinay.amatya@pnnl.gov}

\begin{document} 
\begin{abstract} 
In this paper, we present {\em GossipGraD} -- a gossip communication protocol
based Stochastic Gradient Descent (SGD) algorithm for scaling Deep Learning
(DL) algorithms on large-scale systems. The salient features of GossipGraD are:
1) reduction in overall communication complexity from $\Theta(\log(p))$ for
$p$ compute nodes in well-studied SGD to $O(1)$, 2) model diffusion such that
compute nodes exchange their updates (gradients) indirectly after every
$\log(p)$ steps, 3) rotation of communication partners for facilitating direct
diffusion of gradients, 4) asynchronous distributed shuffle of samples during
the feedforward phase in SGD to prevent over-fitting,  5) asynchronous
communication of gradients for further reducing the communication cost of SGD
and GossipGraD.  We implement GossipGraD for GPU and CPU clusters and use
NVIDIA GPUs (Pascal P100) connected with InfiniBand, and Intel Knights Landing
(KNL) connected with Aries network. We evaluate GossipGraD using well-studied
dataset ImageNet-1K ($\approx$ 250GB), and widely studied neural network
topologies such as GoogLeNet and ResNet50 (current winner of ImageNet Large
Scale Visualization Research Challenge (ILSVRC)).  Our performance evaluation
using both KNL and Pascal GPUs indicates that GossipGraD can achieve perfect
efficiency for these datasets and their associated neural network topologies.
Specifically, for ResNet50, GossipGraD is able to achieve $\approx$ 100\% compute
efficiency using 128 NVIDIA Pascal P100 GPUs -- while matching the top-1 classification accuracy
published in literature.
\end{abstract}

\maketitle

\section{Introduction}
\label{sec:intro}
Deep Learning (DL) algorithms are a class of Machine Learning and Data Mining
(MLDM) algorithms.  A deep neural network (DNN) -- which stores the model of a
DL algorithm -- contains several {\em layers} of {\em neurons} inter-connected
with {\em synapses}. By using a large number of
layers, DL algorithms are able to conduct transformations on highly non-linear
data which is a common property in many science datasets.  DL algorithms have
demonstrated remarkable improvements in many Computer Vision
tasks~\cite{NIPS2012_4824,43022} and science domains such as High Energy
Physics~\cite{Baldi:2014kfa}, and Climate Modeling~\cite{liu2016application}.
Several DL implementations such as TensorFlow~\cite{tensorflow2015-whitepaper},
Caffe~\cite{jia2014caffe}, Theano~\cite{bergstra+al:2010-scipy,
Bastien-Theano-2012}, and Torch~\cite{Collobert02torch:a} have become
available. 

At the same time, DL algorithms are undergoing a tremendous revolution of their
own. Several types of DL algorithms -- primarily geared by input properties
(tabular, images, speech) -- are used by researchers/practitioners.  As an
example, Multi-layer Perceptrons (MLPs) are widely used for tabular data.
Similarly for images, videos and speech, Convolutional Neural Networks (CNNs)
and Recurrent Neural Networks (RNNs) are the {\em de facto} choice. However,
CNNs and RNNs are very computationally expensive, requiring days of training
time using GPUs even on modest size datasets.  Their computational requirements
are further worsened by: 1) very deep neural networks such as
GoogLeNet~\cite{inception} and Residual Networks ({\em
ResNet})~\cite{he2016deep}(Figure~\ref{fig:dnn}), and 2) an increasing volume of data produced by
simulations, experiments and handheld devices.  

An important solution to these problems is distributed memory DL algorithms
that are capable of execution on multi-device (such as multi-GPUs) and large
scale systems. Table~\ref{table:newcomparisons} shows the prominent approaches
for distributed memory DL implementations. 
We observe that the HPC ready implementations
typically incur $\Theta(\log(p))$ communication complexity, while a few of the
HPC ready implementations support overlap of communication with computation. 

\begin{figure}[hptb]
\centering
\small
\includegraphics[width=0.9\columnwidth]{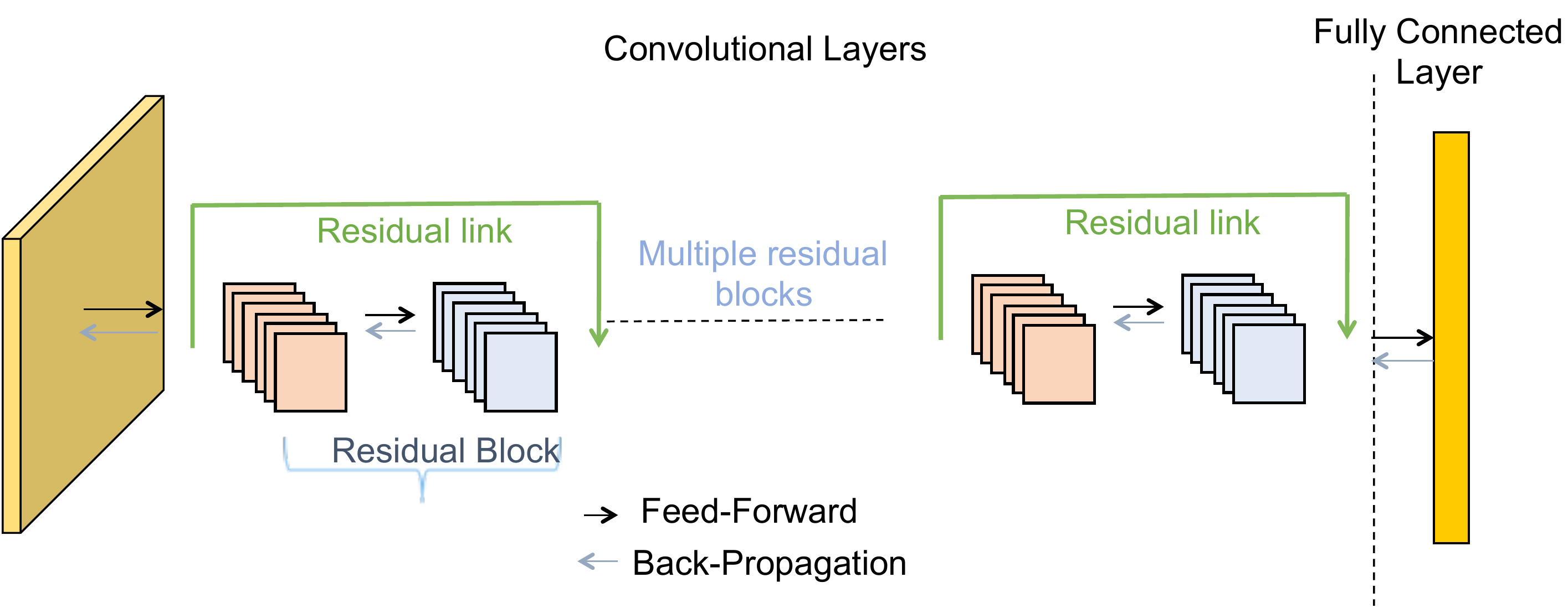}
\caption{\small 
An example of a neural network -- ResNet~\cite{he2016deep}. ResNet has two types of layers -- convolutional and fully connected. Besides feedforward and backprop links, it has a residual link. The residual block is a unit which is repeated to create very deep topologies}. 
\label{fig:dnn}
\end{figure}

DL algorithms primarily use the {\em Gradient Descent} method, which is an
iterative technique to update the weights of {\em synapses} using the
{\em error} between the {\em ground truth} (actual value) and the {\em
predicted value} (using the current state of the DNN). The larger the
difference, the steeper GD converges to a minima -- a low value of minima
generates the solution.  
An important type of Gradient Descent is
Batch/Stochastic Gradient Descent (SGD) -- where a random subset of samples are
used for iterative {\em feed-forward} (calculation of predicted value) and {\em
back-propagation} (update of synaptic weights). 
%

Consider a strong scaling scenario where a batch ($b$) of the input dataset is
split across multiple compute nodes ($p$) -- an example of {\em data
parallelism}.  Under such a scenario, the computation is bounded by
$\Theta(\frac{b}{p})$, although the communication -- requiring an all-to-all
reduction -- is bounded by $\Theta(\log(p))$. With weak scaling -- the work per
compute node remains constant, but the communication increases by $\log(p)$.
This argument is further validated by Goyal {\em et al.}~\cite{goyal:arxiv17}
and PowerAI~\cite{powerai}, since weak scaling
beyond it results in accuracy loss.  Additionally, in practice the actual
communication time deviates from $\Theta(\log(p))$ due to system issues as
pointed out by other researchers~\cite{hoefler:sc10,bhatele:sc13}.

\begin{table}[!htbp]
	\small
	\centering
	\begin{tabular}{|p{2cm}| p{1cm}| p{1cm}| p{1cm}| p{1cm}|}
		\hline
		Name  & Comm. Complexity & HPC Ready  & Comm. Overlap & Leverage Bisection BW\\
		\hline 
		FireCaffe~\cite{firecaffe} & $\Theta(\log(p))$ & \checkmark & \ding{55} & \checkmark \\
		S-Caffe~\cite{scaffe} & $\Theta(\log(p))$ & \checkmark & \checkmark & \checkmark \\
		MaTEx~\cite{matex} & $\Theta(\log(p))$ & \checkmark & \ding{55} & \checkmark \\
		Poseidon~\cite{poseidon} & $O(1)$ & \ding{55} & \checkmark & \ding{55} \\
		Petuum~\cite{petuum} &  $O(1)$ & \ding{55} & \ding{55} & \ding{55} \\
		GeePS~\cite{geeps} &  $O(1)$ & \ding{55} & \checkmark & \ding{55} \\
		ProjectAdam~\cite{projectadam} & $O(1)$ & \ding{55} & \checkmark & \ding{55} \\
		TensorFlow~\cite{tensorflow2015-whitepaper} & $O(1)$ & \ding{55} & \ding{55} & \ding{55} \\
		MXNET~\cite{mxnet} & $O(1)$ & \checkmark & \ding{55} & \ding{55} \\
		CaffeonSpark~\cite{caffeonspark} & $O(1)$ & \ding{55} & \checkmark & \ding{55} \\
		SparkNet~\cite{sparknet} & $O(1)$ & \ding{55} & \checkmark & \ding{55} \\
		CNTK~\cite{cntk} & $\Theta(\log(p))$ & \checkmark & \ding{55} & \checkmark  \\
		Parle~\cite{parle} &  $O(1)$ & \ding{55} & \ding{55} & \ding{55} \\
		PaddlePaddle~\cite{paddlepaddle} &  $O(1)$ & \checkmark & \checkmark & \ding{55} \\
		DistBelief~\cite{dean:nips12} &  $O(1)$ & \checkmark & \ding{55} & \ding{55} \\
		Caffe2~\cite{goyal:arxiv17} &  $\Theta(\log(p))$ & \checkmark & \checkmark & \checkmark \\
		DSSTNE~\cite{dsstne} & $O(1)$ & \checkmark & \ding{55} & \ding{55} \\
		Chainer~\cite{chainer} &  $\Theta(\log(p))$ & \checkmark & \ding{55} & \checkmark \\
		PowerAI~\cite{powerai} &  $\Theta(\log(p))$ & \checkmark & \checkmark & \checkmark \\
		EASGD3~\cite{you:sc17} &  $O(\log(p))$ & \checkmark & \checkmark & \ding{55} \\
		Blot {\em et al.}~\cite{blot:arxiv16} & $O(1)$ & \ding{55} & \checkmark & \ding{55} \\
		\hline
		Others~\cite{ho:nips13,pan:arxiv17,ichinose:imcom17} & $O(1)$ & \ding{55} & \checkmark & \ding{55} \\
		\hline
		{\bf GossipGraD} & {\bf $O(1)$} & {\bf \CheckmarkBold}  & {\bf \CheckmarkBold} & {\bf \CheckmarkBold} \\
		\hline
	\end{tabular}\\
	\caption{\small 
	Comparison of GossipGraD with the major distributed
	Deep Learning approaches. HPC ready implementations leverage the HPC interconnects natively such as either using MPI or native interfaces.
	}
	\label{table:newcomparisons}
\end{table}

Several researchers have proposed methods to alleviate the communication
requirements of distributed
SGD~\citep{NIPS2012_0598,geeps,mxnet,deepimage,li:nips14}. The parameter server
based approaches (shown in Figure~\ref{fig:servergossip}(a))  use a server to hold the latest copy of the model, while
worker(s) send gradients (computed as a function of error) and request the
latest weights.  An important aspect of the parameter server is the expected
constant communication complexity ($O(1)$ instead of $\Theta(\log(p))$ for SGD)
. However, researchers have pointed out the deficiencies of parameter server
based approaches~\cite{li:nips14,scaffe}: 1) a single parameter server becomes
a performance bottleneck, 2) multiple parameter servers (such as using GPUs)
result in waste of compute resources~\cite{scaffe}, 3) communication between
parameter servers in presence of multiple servers is expensive, and 4)
parameter servers require warm-start to facilitate
convergence~\cite{li:nips14}.  Because of those limitations, parameter servers
have only been practical for relatively small clusters, and are
therefore excluded from further consideration in this paper.

Data parallelism based implementations such as
FireCaffe~\cite{firecaffe},S-Caffe~\cite{scaffe}, PowerAI~\cite{powerai},
Caffe2~\cite{goyal:arxiv17}, MaTEx~\cite{matex}, and others~\cite{das:arxiv16}
use all-to-all reduction and other collective operations and use the bisection
bandwidth of communication network effectively (\cref{table:newcomparisons}).
Recently proposed S-Caffe~\cite{scaffe}, PowerAI~\cite{powerai} and
Caffe2~\cite{goyal:arxiv17} also overlap communication with compute
(back-propagation). However, even with these optimizations, communication
becomes a bottleneck at scale even with asynchronicity such as evident in their
evaluation.  

Hence, alternate -- possibly complimentary -- approaches are required to scale
DL implementations further. 
There are two primary approaches published in literature of gossip based DL as
proposed by Jin {\em et al.}~\cite{jin:arxiv16} and Blot {\em et
al.}~\cite{blot:arxiv16}. Under these approaches, a random communication partner is
selected for sending model updates as shown in Figure~\ref{fig:servergossip}(b). These updates are applied by the receiving
compute node and the process is repeated iteratively.  However, both Jin {\em
et al.} and Blot {\em et al.} report performance and convergence degradation at
modest scale.  Up on careful review, we attribute the problems
to: 1) communication imbalance since random selection does not guarantee
balanced communication across compute nodes, 2) communication overhead due to
lack of asynchronicity, and 3) poor convergence at scale due to imbalanced
diffusion of gradients. 
Architecture specific approaches such as proposed by  You {\em et
al.}~\cite{you:sc17} provide support for GPU and KNL architectures, but their
efficiency decreases sharply on using 64 KNL nodes -- which implies that it is
not a feasible solution for larger scale systems.

\begin{figure}[hptb]
\centering
\small
\includegraphics[width=0.7\columnwidth]{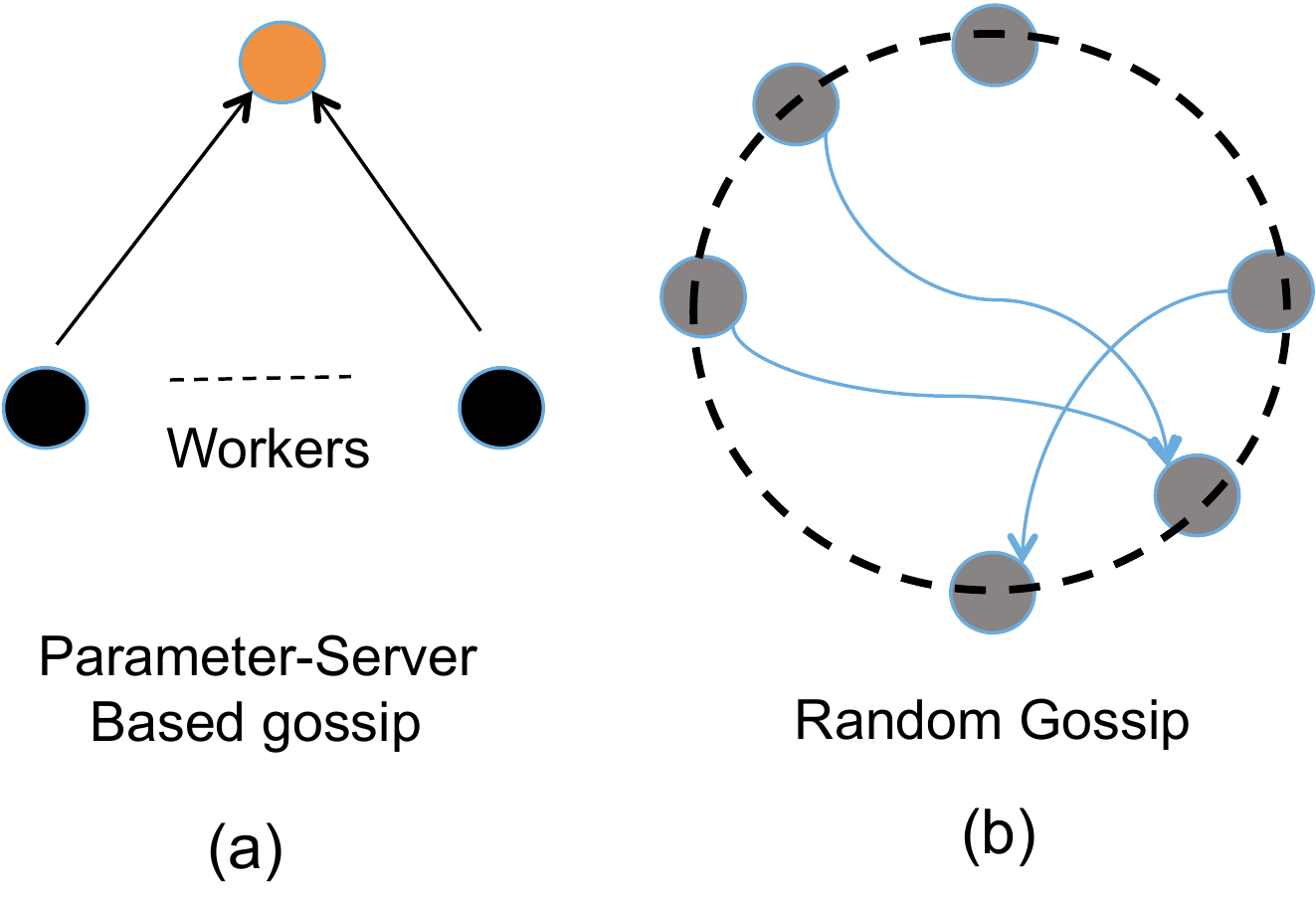}
\caption{\small (a) Parameter-server based implementation -- which is an
extreme form of all-to-one gossip, (b) Random gossip as considered by Blot {\em
et al.}~\cite{blot:arxiv16} and Jin {\em et al.}~\cite{jin:arxiv16}}
\label{fig:servergossip}
\end{figure}
\subsection{Problem Definition and Synopsis}
Existing all-to-all reduction based approaches provide excellent convergence,
but limited scalability. At the same time, parameter server and existing gossip
communication based approaches provide constant communication complexity, but
limited convergence.  The important question is: {\em Can we design a set of distributed memory DL algorithms
that provides constant communication complexity while leveraging asynchronous
communication, but similar convergence properties as sequential SGD?} That is
the objective of this paper.

To achieve this objective, we propose GossipGraD -- a novel gossip
communication protocol based SGD -- to design distributed DL algorithms. Our
algorithm addresses the limitation of existing gossip based approaches by
eliminating communication imbalance, introducing asynchronicity such that the
communication is completely overlapped with computation, and shuffling the
communication partners and the training dataset asynchronously to provide
better diffusion of model updates and preventing over-fitting.


\subsection{Contributions}
Specifically, we make the following contributions in this paper:
\begin{itemize}
\item We present the design choices for communication partner selection
in GossipGraD. Specifically, we present partner selection such
that each compute node communicates with exactly one partner at
each step and after $\log(p)$ steps, all compute nodes have
communicated indirectly.
\item We present techniques for batch-wise rotation of partners for
further enabling direct diffusion of model updates. We consider
layer-wise/batch-wise approaches for asynchronous communication
of model updates to minimize the observed communication time
during the training phase. 
\item To prevent over-fitting (a situation where the model has memorized
the training set), we present {\em asynchronous} distributed memory shuffle of
samples. The shuffle does not incur additional communication
complexity since it is overlapped with the feedforward step.
\item We provide theoretical underpinning of GossipGraD (and its heuristics), 
and present the case that GossipGraD converges to a similar solution as default SGD.
\item We use Caffe -- a widely available toolbox -- as our implementation baseline.
We implement GossipGraD using vendor optimized Caffe --
Intel-Caffe for Knights Landing (KNL) based systems and NVIDIA-Caffe for GPU
based systems. Each of these implementations are also extended
for distributed memory implementation by using MPI~\cite{mpi1,mpi2} such as considered by other researchers including Caffe2, PowerAI, and S-Caffe.
\end{itemize}

We evaluate GossipGraD using the well-studied dataset ImageNet-1K (1.3M images,
$\approx$ 250GB)~\cite{ILSVRC15}, and neural network topologies such as
GoogLeNet~\cite{43022} and ResNet~\cite{he2016deep} using 128 (32 x 4) Pascal
P100 GPUs, and LeNet3 and CIFARNet using 32 node Intel KNL system (Aries
interconnect). Our performance evaluation indicates the effectiveness of
GossipGraD on a variety of architectures. Specifically, we are able to provide
complete overlap of communication with computation for ResNet50 on 128 GPUs
resulting in $> 99\%$ efficiency and match top-1 accuracy
published by other researchers (PowerAI, Caffe2 and Chainer). As an example,
after 30 just training epochs of ResNet50, GossipGraD achieves a top-1 accuracy of
50\% -- which matches with the accuracy published by Caffe2, Chainer and PowerAI.
Our theoretical
underpinnings and accuracy evaluation indicate that GossipGraD provides similar
accuracy as the well-studied SGD algorithm. Hence, GossipGraD is suitable for
deployment on extreme scale systems.

The rest of the paper is organized as follows: In section~\ref{sec:background},
we present the background of the proposed research. We present baseline setup
in section~\ref{sec:design}, GossipGraD design in section~\ref{sec:gossipgrad},
introduce asynchronicity in section~\ref{sec:asyncgd} and theoretical proof of
convergence in section~\ref{sec:proof}. We discuss experimental results in
section~\ref{sec:exp}, related work in section~\ref{sec:related},  followed by
conclusions in section~\ref{sec:conclusions}.

\section{Background}
\label{sec:background}

\subsection{Gradient Descent Algorithm}
The gradient descent algorithm -- which iteratively updates the weights
of a deep neural network to a local minimum of the cost function -- is the most
widely used DL algorithm. 
In the algorithm, each sample (which could be tabular
observation, an image, or speech/text), is an input to the {\em feed-forward}
step. The output of each feed-forward step is a {\em predicted value}, which is
compared with the {\em label} (ground truth), and their difference is used to
calculate the {\em gradients} -- which are applied for updating the weights. 
In
essence, the feed-forward step calculates the {\em loss} -- a measure of
difference between label and predicted values, and the objective is to minimize
the loss on the training set, while ensuring that the validation loss decreases
as well. The general technique is typically referred to as {\em gradient
descent} towards a minimum of the loss function (Figure~\ref{fig:sgd_loss_example}).

\begin{figure}[hptb]
\centering
\includegraphics[width=0.9\columnwidth]{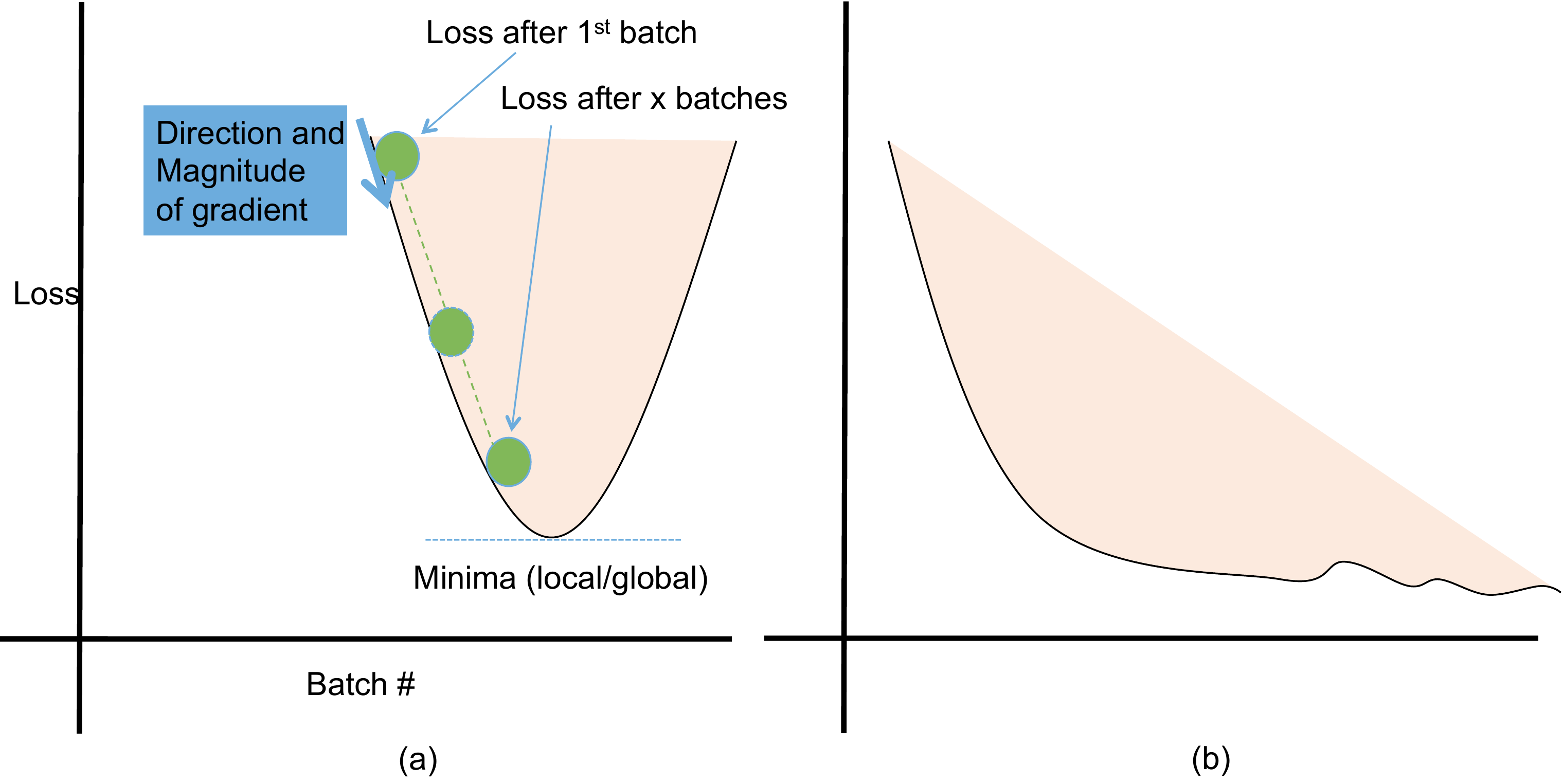}
\caption{(a) A pictorial representation of a loss curve as function of batch -- each batch implies a model update,
(b) A typical DL loss curve containing long-tail and non-convex shape}
\label{fig:sgd_loss_example}
\end{figure}

For gradient descent and its variants, the rule(s) for updating the weights
during the back-propagation step can be presented as follows:
\begin{eqnarray}
	\label{graddesc}
	\mathbf{w}'&=&\mathbf{w}+\lambda \nabla_{\mathbf{w}}C\\
	\mathbf{b}'&=&\mathbf{b}+\lambda \nabla_{\mathbf{b}}C
\end{eqnarray}
where $\mathbf{w}$ are the weights, $\mathbf{b}$ the biases, $\lambda$ the
learning rate, and $C$ is a cost function to be optimized which is usually square error
or cross-entropy.   To compute the derivatives, we set $W^{(\ell)}$,
$b^{(\ell)}$ to be the weights and biases for each layer,
$z^{(\ell+1)}=W^{(\ell)}a^{(\ell)}+b^{(\ell)}$ and
$a^{(\ell)}=\sigma(z^{(\ell)})$, where $\sigma$ is the activation function, and
we set $n_\ell$ represent the number of layers.  The details of the algorithm
are presented in Algorithm~\ref{alg:bp}.

\begin{algorithm}
	\caption{Back Propagation~\cite{marsland2015machine}}
	\begin{algorithmic}[1]
		\label{alg:bp}
		\STATE \textbf{input:} Data $X\in \RR^{n\times p}$ and labels $Y\in \RR^{n\times l}$
		\STATE Compute all $z^{(\ell)}$ and $a^{(\ell)}$.
		\STATE $\delta^{(n_\ell)} = -(y-a^{(n_{\ell})})\odot \sigma(z^{(n_\ell)})$
		\FOR{$\ell$ from $n_\ell-1$ to 1}
		\STATE $\delta^{(\ell)}=W^{\ell} \delta^{(\ell+1)}\odot \sigma'(z^{(\ell)})$
		\ENDFOR
		\STATE $\nabla_{W^{(\ell)}}C = \delta^{(\ell+1)}{a^{(\ell)}}^T$
		\STATE $\nabla_{b^{(\ell)}}C = \delta^{(\ell+1)}$
	\end{algorithmic}
\end{algorithm}

\section{Baseline Setup}
\label{sec:design}
In this section, we present a solution space for scaling SGD on distributed
memory systems. We first present a distributed memory implementation of {\em
Vanilla SGD} (simply referred as SGD for rest of the paper).  SGD provides a
baseline for performance and accuracy comparisons with GossipGraD.
Table~\ref{table:modeling} shows the symbols used for time-space complexity
analysis of the proposed approaches.

\begin{table}[!t] 
\centering
\begin{tabular}{|c|c|c|c|c|}
\hline
& Method & Symbol\\
\hline 
1 & Input Dataset & $X$\\   
2 & Batch size & $b$\\   
3 & Number of Processes & $p$ \\  
4 & Network Latency & $l$    \\
5 & Network Bandwidth & $\frac{1}{G}$    \\
6 & Set of Layers & $L = {L_0 .. L_{n - 1}}$    \\
\hline
\end{tabular}\\  
\caption{Symbols used for GossipGraD and Other Approaches}
\label{table:modeling}  
\end{table}

\subsection{Distributed Memory SGD}
These are several design choices for scaling SGD on distributed memory systems.
Specifically, to maintain equivalence to the sequential algorithm, a strong
scaling approach may be useful. Considering a batch size ($b$), the overall
work for ($p$) processes is expected to be $\Theta\left(\frac{b}{p} +
||L||\log{p}\right)$, where $||L||$ is the overall size of the synaptic weights
in a DNN, and an all-to-all reduction is conducted among all MPI processes.
The $\log(p)$ term is obtained due to the time-complexity of the reduction tree
typically implemented as a binomial/k-nomial tree in distributed systems. For
weak-scaling, we expect the time complexity to be $\Theta(b + ||L||\log{p})$,
since the work/process remains constant.  Figure~\ref{fig:vanilla_sgd} shows an
example of distributed memory SGD with ResNet topology.

\begin{figure}[hptb]
	\centering
	\includegraphics[width=0.8\columnwidth]{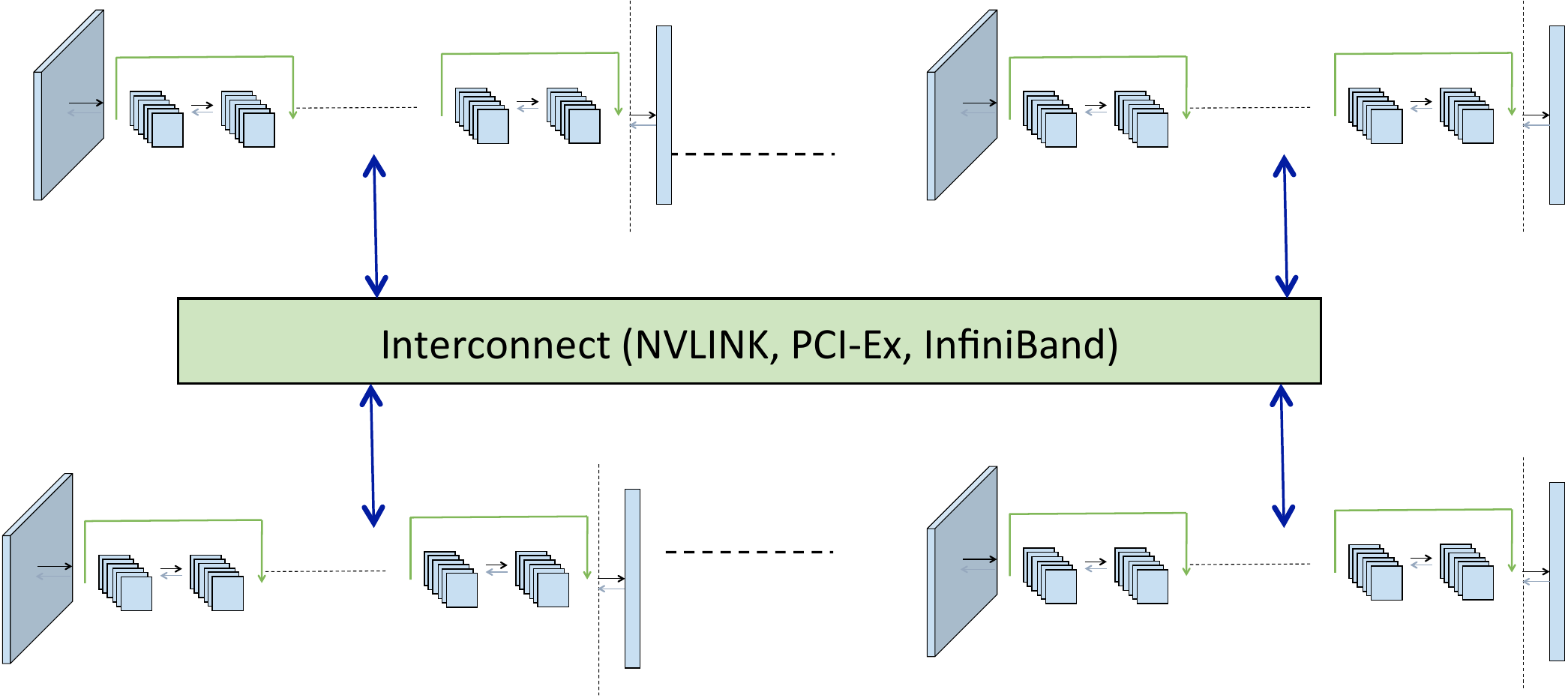}
	\caption{Pictorial representation for SGD using ResNet
	example. Data parallelism is used by replicating the ResNet model
	across all compute nodes. All-to-all reduction is executed after
	back-propagation is complete.}
	\label{fig:vanilla_sgd}
\end{figure}

\subsection{Scaling Limitations of SGD: Possible Solutions}
An advantage of distributed SGD is its strict equivalence to the sequential SGD
algorithm. However, the communication complexity -- bounded by
$\Theta(\log(p))$ -- becomes a limiting factor at scale.  Hence, it is
important to consider alternate techniques to scale SGD on large scale systems.
One possibility is to communicate to fewer processes -- possibly to exactly one
process after each batch update. This reduces the communication complexity from
$\Theta(\log(p))$ to $O(1)$. In section~\ref{sec:gossipgrad}, we consider design
approaches for GossipGraD which achieves this objective.  

Another (complimentary) possibility is to introduce asynchronicity in SGD (AGD)
and GossipGraD. Recently proposed solutions such as S-Caffe~\cite{scaffe},
Caffe2~\cite{goyal:arxiv17} and PowerAI~\cite{powerai} overlap communication
with computation during the back-propagation step.  This is evident from
Algorithm~\ref{alg:bp}, since gradients for each layer are ready for
communication before back-propagation is initiated on preceding layers.  We
present design choices for asynchronous execution in section~\ref{sec:asyncgd}.


\section{GossipGraD Design}
\label{sec:gossipgrad}
In this section, we present a solution space for designing GossipGraD.  We
first present an intuition for GossipGraD.  We consider several elements of
GossipGraD: 1) communication partner selection for exchanging model updates, 2)
techniques for asynchronous data shuffle for preventing over-fitting, 3)
low-overhead partner rotation for better diffusion of model updates and
accelerating convergence to the solution, and 4) implementations on state of the art GPU and CPU clusters. 

\subsection{Extreme Case: No Communication}
Let us consider the basic SGD algorithm presented earlier. After each batch, an
all-to-all reduction (typically implemented using \texttt{MPI\_Allreduce}) is
conducted, which takes $\log(p)$ communication steps for completion.  The
objective of GossipGraD is to reduce the communication complexity of SGD from
$\log(p)$ to $O(1)$. One possibility is to completely eliminate communication
from SGD.  This approach is attractive for minimizing communication complexity.
However, the no-communication approach has two problems: 1) each model (by data
parallelism) learns only on a subset of data.  This implies that the models
drift apart further at increasing scale of compute nodes, and 2) this approach
would produce an ensemble of DL models, while our objective is to produce a
single model at the end of the training phase. Hence, we disregard no communication as a viable solution to addressing the communication complexity of DL algorithms.

\subsection{Intuition for Gossip Communication}
Gossip -- as frequently observed in social networks -- is frequently used in
computer-to-computer communications especially in distributed systems.  Compute
nodes may select a random partner for gossip, and each exchange is not expected
to be reliable.  However -- over a period of time -- gossip communication is 
expected to provide robust information across partners.  This property of
gossip communication protocols is the premise of GossipGraD.
As each pair of compute nodes gossip their model updates -- with a potentially
different partner at each step --  the expectation is that individual models of
compute nodes steadily converge towards the same model (theoretical
underpinnings are presented in section~\ref{sec:proof}). SGD usually takes
several epochs (as an example ResNet requires 90 epochs), where each epoch may
have thousands of batches.  The premise of GossipGraD is that by {\em
gossiping} between compute nodes over thousands of batches, the model updates
are diffused across all compute nodes.

Figure~\ref{fig:servergossip} demonstrated the limitations of existing gossip
based approaches, including the parameter server approaches, which can be
considered as an extreme form of gossip.  The approaches presented by Jin {\em
et al.}~\cite{jin:arxiv16} and Blot {\em et al.}~\cite{blot:arxiv16} which use
random gossip suffer from communication imbalance, poor diffusion of gradients,
lack of data shuffle and asynchronous exchange of gradients.  The proposed
GossipGraD design intends to address each of these limitations. We begin with a
discussion on selection gossip communication partner.

\subsection{Gossip Communication Partner Selection}
The approach of selecting a communication parter should have the following
properties: 1) constant communication complexity -- such that the solution is
scalable to very large scale systems, 2) balanced communication such that at
each step there is a unique set of communicating nodes, 3) diffusion of
gradients in sub-linear time, and 4) leveraging the bisection bandwidth of the
communication network.  

To facilitate constant communication complexity and balanced
communication, we propose a hierarchical virtual organization of compute nodes
such as in hypercube topology (shown in Figure~\ref{fig:gossip_approaches}(a)).
Other configurations such as dissemination-based (shown in
Figure~\ref{fig:gossip_approaches}(b)) have similar property as hypercube
virtual topology. However, dissemination-based approach is better since it
communicates with exactly two compute nodes at each step. Another important
property of hypercube based and dissemination based approaches is that all
compute nodes have communicated indirectly with each other in exactly $\log(p)$
batches -- which achieves our objective of sub-linear diffusion of model updates.

\begin{figure}[hptb]
	\centering
	\includegraphics[width=1.0\columnwidth]{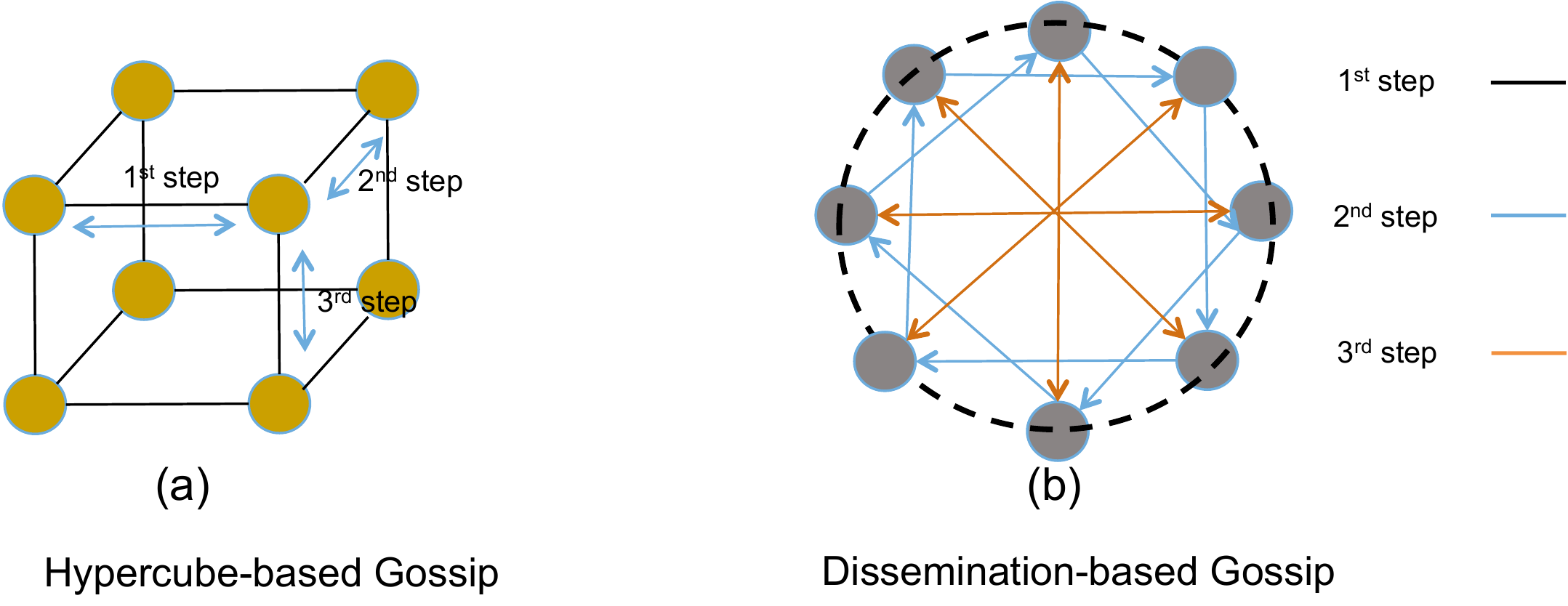}
	\caption{(a) Hypercube based virtual topology organization of compute nodes (b) dissemination based organization of compute nodes}
	\label{fig:gossip_approaches}
\end{figure}
\subsection{Expected diffusion of gradients with Hypercube and Dissemination Virtual Topologies}
\subsubsection{Diffusion in Hypercube Topology}
An example is shown in Figure~\ref{fig:gossipgrad_intuition}.  Intuitively, the
expectation is that by exchanging gradients with a partner, the gradients are diffused {\em indirectly} over a period of steps. As evident from the
figure, after $\log(p)$ steps, the gradients have been indirectly diffused with
all compute nodes. 

\begin{figure}[hptb]
\centering
\includegraphics[width=1.0\columnwidth]{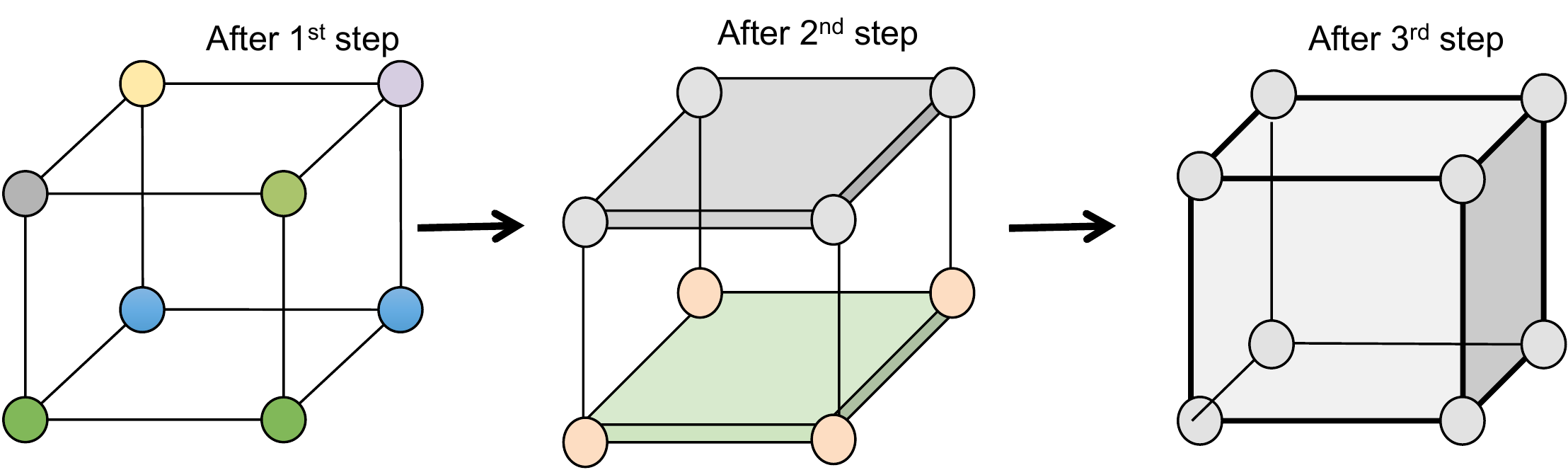}
\caption{
The expected diffusion of gradients
across steps. After 1st step, set of pairs
have same gradient, after 2nd step, top
and bottom surfaces, and after 3rd step, 
the entire set of processes.}
\label{fig:gossipgrad_intuition}
\end{figure}


\subsubsection{Diffusion in Dissemination Topology}
The primary difference between dissemination and hypercube 
algorithms is the partner selection. At step $k$ of the all-to-all reduction, a
process $p_i$ sends data to a process with MPI rank $(p_i + 2^k) \% p $, and receives
from $(p_i + p - 2^k) \% p $.
Figure~\ref{fig:dissemination} shows the steps in classical all-to-all
reduction and GossipGraD, when using the dissemination algorithm.  Similar to
hypercube exchange, dissemination provides $\log(p)$ time complexity. 
\begin{figure}[hptb]
\centering
\includegraphics[width=0.7\columnwidth]{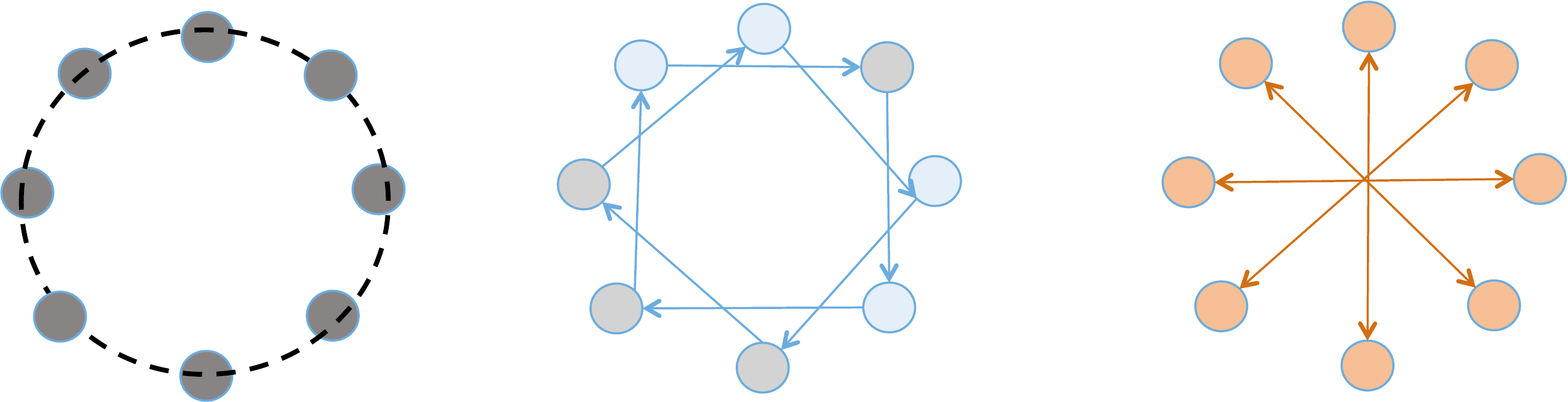}
\caption{An example of exchanges in dissemination algorithm for GossipGraD. 
Unlike hypercube exchange, each process sends and receives from a different partner and completes indirect exchange in $\log(p)$ steps}
\label{fig:dissemination}
\end{figure}

However, for GossipGraD, the difference between dissemination and hypercube
based partner selection is worth consideration.  Specifically, by using the
dissemination algorithm, GossipGraD is diffusing gradients from two partners at
each step, while hypercube exchange is diffusing from exactly one partner.
Hence, we primarily
consider dissemination exchange based partner selection for GossipGraD.

\subsection{Handling Side-Effects of Reduced Communication}
A possible side-effect of the reduced communication is divergence from the
sequential algorithm -- since only a small subset of compute nodes are
communicating at each step. We propose two techniques to address this:

\subsubsection{Partner Rotation}
Let us consider an execution of dissemination exchange based GossipGraD. We observe
that the communication partners are repeated after $\log(p)$ steps. As a
result, the direct diffusion of gradients is restricted to a small fraction
($\frac{\log(p)}{p}$) partners.  To alleviate this limitation, we propose a
{\em partner rotation} based approach. The premise of the rotation based
approach is such that the communication partners are
modified after every $\log(p)$ steps. The rotation based approach facilitates
direct diffusion of gradients to all compute nodes over a period of time,
without increasing the time complexity of GossipGraD.  
To achieve this objective, we consider $p$ random shuffles of the original 
communicator. After each $\log(p)$ steps, the next shuffled communicator is
used for creating a new virtual dissemination topology. Since the communicators 
are created at start of the application, the overall cost of creating them is
easily amortized over long running training phase of the DL implementation.


\subsubsection{Rotation of Samples}
The majority of distributed memory DL implementations read the samples once in
memory at the beginning of the training phase, and use them repeatedly till
convergence. This approach is sufficient for default SGD implementation --
since all compute nodes communicate after every batch. Since GossipGraD reduces
the communication complexity, we propose a distributed {\em shuffle} of the
samples such that after a compute node has used a batch of samples, it sends
the samples to another compute node. To develop a different topology from the
dissemination topology of gradients, we consider a ring virtual topology for
the shuffle of samples, where each compute node sends a recently completed
batch to its neighbor. A simple yet effective shuffle provides a nice property that each sample is considered again for feedforward/backpropagation step 
by a compute node only after all other compute nodes have considered it once.



\section{Asynchronous GossipGraD}
\label{sec:asyncgd}
As discussed in the previous section, we have reduced the communication
complexity from $\Theta(\log(p))$ to $O(1)$. However, we have added the shuffle of
samples -- potentially increasing the overall communication observed by the
compute nodes at the training phase. Another important consideration is
leveraging asynchronous communication for exchanging model updates. We leverage
the property of feedforward and backpropagation step to asynchronously shuffle
the samples and exchange the model updates.

An important observation from Algorithm~\ref{alg:bp} and GossipGraD design is
that the gradients for each layer are ready for diffusion before the gradients
for preceding layers are computed. Hence, it is possible to overlap the
computation of gradients for preceding layers by communicating gradients of the
current layer. Existing approaches such as S-Caffe~\cite{scaffe},
PowerAI~\cite{powerai} and Caffe2~\cite{goyal:arxiv17} have made similar
observation for all-to-all reduction. We use a similar property for
point-to-point communication in GossipGraD.

\subsection{Distributed Asynchronous GossipGraD using MPI Non-Blocking Primitives}
For GossipGraD, we use
non-blocking point-to-point communication primitives --\texttt{MPI\_Isend} and
\texttt{MPI\_Irecv}.  We implement GossipGraD by initiating non-blocking
send/receive, as soon as the gradients for a layer are ready.
Figure~\ref{fig:palagrad_nbcl} shows a pictorial representation. The
non-blocking MPI requests return an MPI handle. To maximize the available
overlap, our implementation generates the communication requests and executes
\texttt{MPI\_TestAll} followed by an \texttt{MPI\_WaitAll} function, after the gradients for all the layers have
been computed.
\begin{figure}[hptb]
\centering
\includegraphics[width=\columnwidth]{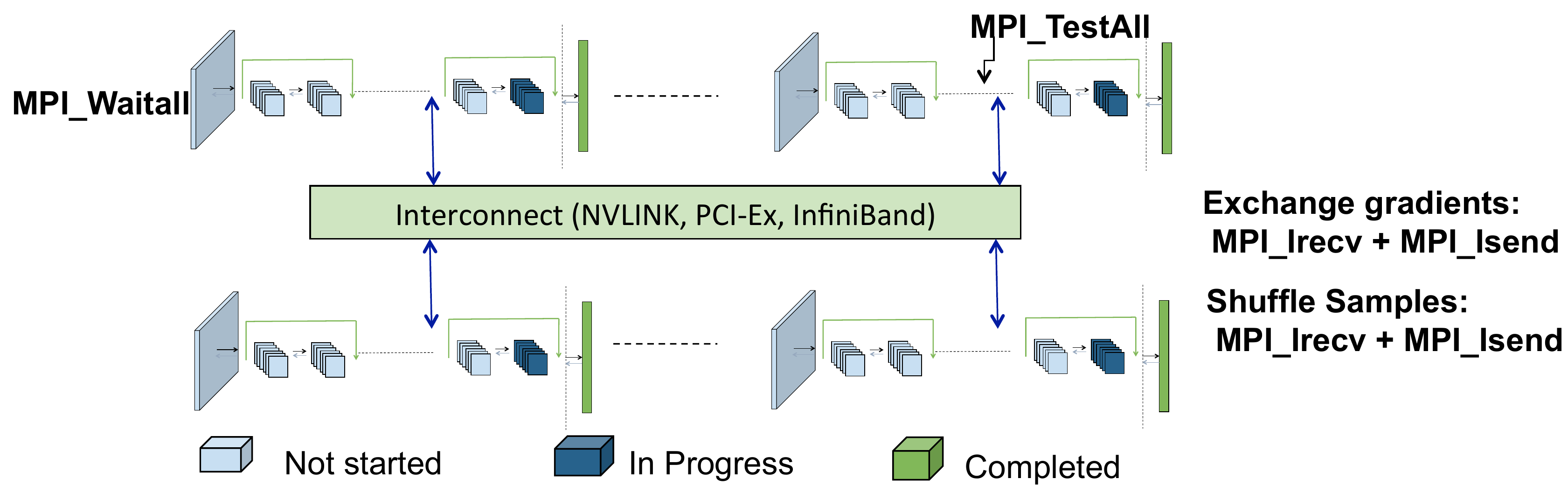}
\caption{A representative execution of AGD and GossipGraD using non-blocking MPI primitives}
\label{fig:palagrad_nbcl}
\end{figure}

\subsection{Distributed Asynchronous GossipGraD using Asynchronous Progress Thread}
A critical problem with non-blocking point-to-point and collective operations
is that it requires MPI runtimes to make asynchronous progress/hardware
support. However, this feature is not necessarily available in all implementations.  At
the same time, blocking operations (such as \texttt{MPI\_Allreduce}) are
heavily optimized by vendors, since they are used in many large scale
applications.  
Let us consider ResNet50 -- which has about 25M parameters (100 MBytes of data) -- to
motivate the need for an asynchronous thread based implementation.  Many of the
layers communicate large data ($>$ 1MBytes) -- which is much greater than the
size of {\em rendezvous threshold} in MPI implementations. Hence, we expect
that the majority of communication uses rendezvous protocol for
point-to-point communication, which necessitates asynchronous communication
progress. Sur {\em et al.} have proposed MPI implementations to support
asynchronous progress~\cite{sur:ppopp06} -- although they may not be available
in practice. 

Hence, we use an asynchronous thread to ensure progress on blocking
point-to-point and collective communication primitives.  Since each message
transfer is sufficiently large to saturate the network bandwidth, a single
communication thread is sufficient.  Once the gradients for a layer are
generated, the main compute thread {\em enqueues} a pointer to the gradients on
a {\em communication queue}. The asynchronous thread dequeues from the
communication queue and issues a blocking point-to-point/all-to-all reduction
operation. Figure~\ref{fig:palagrad_design} provides a pictorial representation
of the asynchronous implementation.

\begin{figure}[hptb]
	\centering
	\includegraphics[width=\columnwidth]{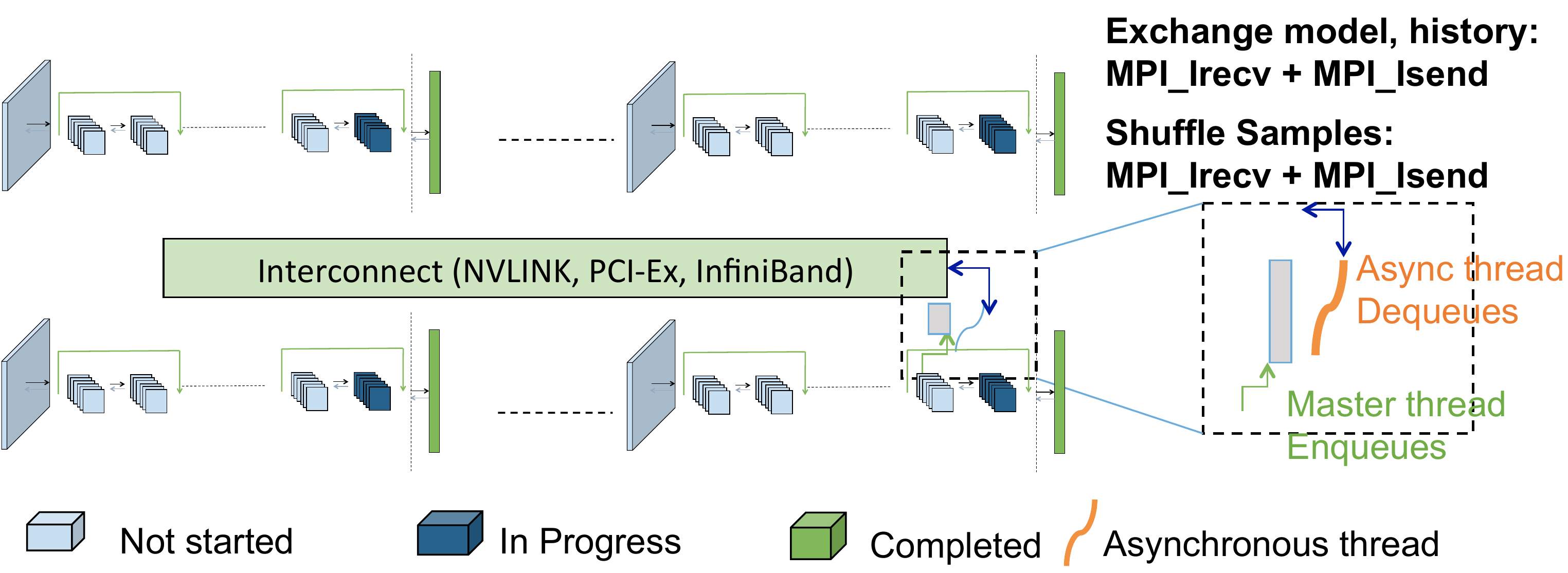}
	\caption{GossipGraD implementation using blocking collectives with
	asynchronous communication threads. Main thread enqueues the gradients,
	which are dequeued by the communication thread}
	\label{fig:palagrad_design}
\end{figure}

\subsubsection{Considerations for Implementation}

There are several important considerations in this implementation: 1) The
thread contention is insignificant, since the master thread enqueues the
gradients, and a single asynchronous thread, which dequeues for
point-to-point/all-to-all reduction, 2) the implementation requires
multi-threaded MPI which is available on most platforms, 3) presence of an
asynchronous thread does not affect the performance of GPU based implementation
(since thread is scheduled on CPUs), for CPU based implementation a single
communication thread negligibly affects the available parallelism. 

Yet, in practice the performance of multi-threaded MPI is not consistent across
all platforms. Specifically, we have observed a performance degradation when
\texttt{MPI\_THREAD\_MULTIPLE} is used, and depending up on the implementation,
an asynchronous progress may not be provided. Fortunately, using
\texttt{MPI\_TestAll} -- which is a non-blocking operation to invoke progress
engine in MPI -- we have observed an expected overlap of communication with
computation. Hence, we discard the asynchronous thread based implementation and
use the non-blocking primitive and testall based implementation for communication.

\section{Sketch of Proof of Convergence} 
\label{sec:proof}

\begin{table}[!t] 
\centering
\begin{tabular}{|c|c|c|c|c|}
\hline
& Meaning & Symbol\\
\hline 
1 & Weights and Biases & $w$\\   
2 & Cost Function & $C$ \\  
3 & $i$th sample & $x_i$ \\
4 & $i$th label & $y^L_i$ \\
5 & predicted value for $i$th sample & $y_i$ \\
\hline
\end{tabular}
\caption{Symbols used for Proof}
\label{table:proofsym}
\end{table}

In this section, we provide a sketch of a proof that the GossipGraD algorithm
converges to a similar model as SGD at the end of training phase.  We 
use the symbols described in Table~\ref{table:proofsym}.
Recall that for a classification network. the \textit{cost function} is the
cross-entropy function defined by \[C(w)=-\frac{1}{n}\sum_{i=1}^n
y_i^L\cdot\log(y_i)\] 

Let the number of nodes be $p$ and assume that the data is distributed between nodes.

\begin{lemma}
\label{lemma:cost}
With shuffling, the compute nodes have the same cost function.
\end{lemma}

\begin{proof}
Without shuffling, the cost functions are summations over the samples resident
on the node.  However, shuffling ensures that over time, each sample is only
re-considered for feedforward and backpropagation once the compute node has considered every
other sample. Hence, the cost function being optimized is the summation
over all samples, and thus identical across nodes.
\end{proof}

Due to the structure of neural networks, the cost function $C(w)$ is non-convex
(this non-convexity is analyzed in detail in~\cite{dauphin2014identifying}).
The non-convexity implies that there are several possible local minima even
when executed on a single device. Even though the cost functions are same
across all compute nodes, they may end up at different local minima due to the
stochastic nature of SGD. Hence, communication of model updates is imperative to guarantee
convergence to a single local minimum. Naturally, this problem becomes worse at scale.

In section~\ref{sec:design}, we have presented communication approaches where
after each batch, the model on each compute node is averaged with another
unique compute node.  
Using the communication approaches presented earlier, we obtain:

\begin{theorem}
\label{thm:main}
Each compute node in GossipGraD converges to a local minimum of the cost function.
\end{theorem}

\begin{proof}
The key technical result is~\cite{supermartingale_conv} which shows that for
any process (called a \textit{positive supermartingale}) where the expectation
value of the ${n+1}^{th}$ batch, conditioned on the first $n$ batches, is
nonnegative and at most the value of the $n$th batch, converges to a
limit with probability 1.  This result has been extended
in~\cite{saad1998online, kiwiel2001convergence} to show that SGD converges
(with probability 1) for any cost function that is continuous and which is
differentiable at all but finitely many points.

By Lemma~\ref{lemma:cost}, each compute node is optimizing a copy of the same objective
function to find a local minimum.  We follow the standard proof for SGD, where
our goal is to show that the sequence $\{C^k=C(w_k)\}_{k=0}^\infty$ is a
positive supermartingale.  Positivity follows from the structure of $C(w)$, an
average of explicitly positive terms, such as the negative of the logarithms of probabilities.
We use the properties of SGD described below to show that the expectation value of the cost is at most the
value of the cost in a given batch.

The properties of SGD that are typically used to show that it is a supermartingale are:
1) the weight differences are bounded, 2) the set of non-minima where the
gradient vanishes is encountered with probability zero, and 3) the expectation
of the difference in weights before and after each batch points towards a local
minimum (denoted by $w^*$), that is, \[\mathbb{E}\left[(w_{n+1}-w_n)\cdot
(w_{n+1}-w^*)\right]\geq 0.\]

The first two properties hold for GossipGraD because they hold for SGD, so only
the third property remains.  Here, the difference between GossipGraD and SGD is
that $w_{n+1}$ for SGD is defined by global gradient update, but $w_{n+1, j}$ for
GossipGraD on compute node $j$ is defined by $(W_{n+1,j} + W_{n+1,c_i(j)})/2$ where
$W_{i,j}$ are the weights at batch $i$ on compute node $j$ after gradient updates and
$c_i$ is the permutation at batch $i$ determining which other compute node is being
averaged into compute node $j$.

As $n$ increases, the probability that this expectation value is negative
decreases, such that for any $\epsilon>0$, there exists an $N$ such that for
all $n>N$, we have \[\mathrm{Pr}(\mathbb{E}\left[(w_{n+1}-w_n)\cdot
(w_{n+1}-w^*)\right]< 0)<\epsilon.\]  

The desired inequality does not hold only in pathological cases, where a
compute node $j$ has consistently moved away from the local minimum.  As in the
case of SGD, however, the expectation of the gradients points towards the local
minimum.  Hence, the desired inequality holds with probability 1 as the number
of batches increase, such as observed in practical cases of large datasets and
DNNs.
\end{proof}

\begin{corollary}
All compute nodes converge to the same local minimum of the cost function.
\end{corollary}

\begin{proof}
Theorem~\ref{thm:main} states that the model on each compute node converges to
some local minimum.  We prove that they converge to the same local minima by
using contradiction.  Let us assume that all models have converged, but not to
the same local minimum.  Then, for some batch, models that are at different
local minimum are averaged together due to partner rotation.  However, this
would cause the weights of these models to change, contradicting the assumption
that they have already converged.  Therefore, all models must have converged to
the same local minimum.
\end{proof}

\section{Performance Evaluation}
\label{sec:exp} 
In this section, we present an in-depth evaluation of GossipGraD and its
associated heuristics using datasets such as MNIST, CIFAR10 and well-studied neural network topologies on
ImageNet-1K dataset. The table~\ref{table:arch} provides a description of the
architectures used for performance evaluation.  Table~\ref{table:datasets}
provides a brief description of the datasets used for performance evaluation.
Table~\ref{table:heur_desc} provides a description of the heuristics.

\begin{table*}[!htbp] 
		\centering
		\begin{tabular}{|c|c|c|c|c|c|c|c|c|}
				\hline
				Name  & CPU (\#cores) & \# GPUs/node & Baseline & Network & MPI &  Nodes & \#cores & \# GPUs\\
				\hline 
				\bf{P100} & Power8 (20) & NVIDIA Pascal P100 (4) & NVIDIA-Caffe\cite{caffe-nvidia} & NVLink, IB-EDR  & IBM MPI & 32 & N/A & 128\\
				\bf{KNL} & Intel KNL (68) & N/A & Intel-Caffe\cite{caffe-intel} & Cray Aries & Cray MPICH & 32 & 2176 & N/A\\
				\hline
		\end{tabular}\\  
		\caption{A description of system architecture and associated software}
		\label{table:arch}  
\end{table*}

\begin{table*}[!t] 
\centering
\begin{adjustbox}{max width=\textwidth}
\begin{tabular}{|c|c|c|c|c|c|c|c|c|}
\hline
Dataset  & Neural Network     & Description & Training Samples & Validation Samples & Image Size            & Classes\\
\hline 
MNIST~\cite{mnistlecun}    & LeNet3~\cite{lecun1998gradient} & Handwritten Digits & 60000            & 10000          &     $28 \times28$ & 10        \\
CIFAR-10~\cite{Krizhevsky09learningmultiple} & CIFARNet  & Small Images & 60000            & 10000              & $32 \times32\times3$  & 10  \\
ImageNet & GoogLeNet~\cite{43022} & Diverse Images & 1281167          & 50000              & $256\times256\times3$ & 1000  \\
ImageNet & ResNet~\cite{he2016deep} & Diverse Images & 1281167          & 50000              & $256\times256\times3$ & 1000\\
\hline
\end{tabular}
\end{adjustbox}
\caption{A description of datasets and associated neural network topologies}
\label{table:datasets}  
\end{table*}

\begin{table*}[!t] 
		\centering
        \small
		\begin{tabular}{|c|c|c|l|}
				\hline
				Name  & Type  & Implemented & Description of DL Algorithm and Implementation\\
				\hline
				AGD & Asynchronous GD & Yes & Implements SGD by asynchronous layer-wise communication. \\
				GossipGraD & Batch-wise Gossip GD with Rotation of partners and samples& Yes & Implements batch-wise AGD using Gossip and Rotation of Samples \\
				\hline
		\end{tabular}\\  
		\caption{A description of GossipGraD Approaches. We implement these approaches and evaluate them on GPU and CPU systems.}
		\label{table:heur_desc}  
\end{table*}

\subsection{Baseline Setup}
DL algorithms have a few hyperparameters which are setup at the start of the
training phase. These include the batch size, learning rate and momentum. We
fix the momentum to be exactly same as provided by default versions of
Intel-Caffe and NVIDIA-Caffe. Since we use weak-scaling, the effective batch
size is multiplied by the number of devices (such as number of GPUs or number
of compute nodes).
For handling weak-scaling with SGD/AGD, we use suggestions by
Krizhevsky~\cite{oneweirdtrick} that multiplies the learning rate on a single device by $\sqrt(p)$.
For GossipGraD, we use
the same batch size on each compute node as defined for a single device (such
as a single GPU or a single KNL compute node), and keep the learning rate unchanged. 
We use AGD -- a theoretically equivalent asynchronous gradient descent
implementation -- as suggested by PowerAI, Caffe2 and Chainer as the baseline
for performance and accuracy comparisons.  We consider AGD to be a better baseline than
default SGD since AGD provides better performance than the SGD, and provides theoretical equivalence to SGD.

\subsection{Comparing Gossip to AGD: Evaluation on (LeNet3) MNIST and (CIFARNet) CIFAR10}
\subsubsection{Performance Results}
Figures~\ref{fig:MNIST-speedup} and~\ref{fig:CIFAR10-speedup} show the
relative speedup of GossipGraD on MNIST and CIFAR10 datasets to AGD implementation on KNL and
P100 clusters respectively.  MNIST on 32 GPUs needs 1.2s per epoch (batch size
on each device is 64), while CIFAR10 on 32 GPUs needs 0.75s per epoch (batch
size on each device is 100).
Both MNIST and CIFAR10 are relatively small
datasets, but provide a delicate balance between communication (since network
sizes are small) and computation (since compute is relatively small). We
observe the following: 1) relative speedup of GossipGraD to AGD for P100 on both
MNIST and CIFAR10 is higher in comparison to KNL, since a single P100 GPU is
much faster than single KNL node, and 2) even with weak-scaling, the relative
speedup continues to increase since GossipGraD is able to overlap the
communication effectively due to $O(1)$ communication, while AGD is unable to
do so even in the presence of a layer-wise asynchronous communication.

\subsubsection{Accuracy Results}
Figure~\ref{fig:MNIST} shows the accuracy charts for GossipGraD on KNL,
GossipGraD on P100 and AGD implementation as a function of number of epochs. We
use 32 P100 GPUs (8 compute nodes) and 32 KNL compute nodes for the evaluation
to validate whether the accuracies match at the largest scale of evaluation on
MNIST considered in this paper.  We observe that the validation accuracy
saturates to $\approx$ 99.2\% for three implementations -- which empirically
proves the argument that GossipGraD provides similar accuracy as AGD
implementation. Figure~\ref{fig:CIFAR10} compares the accuracy charts for
GossipGraD on KNL, GossipGraD on P100 and AGD implementation using 32 KNL nodes
and 32 GPUs on CIFAR10 dataset. The three implementations roughly track each
other and converge towards similar accuracy ($\approx$ 72-73\%), which is
within the margin of error.

\subsubsection{Lessons Learned}
We observe that Gossip provides similar accuracy while providing significant
speedup in comparison to the AGD implementation even on relatively small scale
KNL and P100 clusters. We observe about 1.9x speedup for MNIST since GossipGraD
is able to overlap the communication of both samples and model updates with the
feedforward and backpropagation step. We also conclude that GossipGraD is able to completely overlap the communication -- which implies that it is useful for small size datasets as well.

\begin{figure*}[htbp]
\begin{minipage}[t]{0.66\columnwidth}
\centering
\includegraphics[width=\columnwidth]{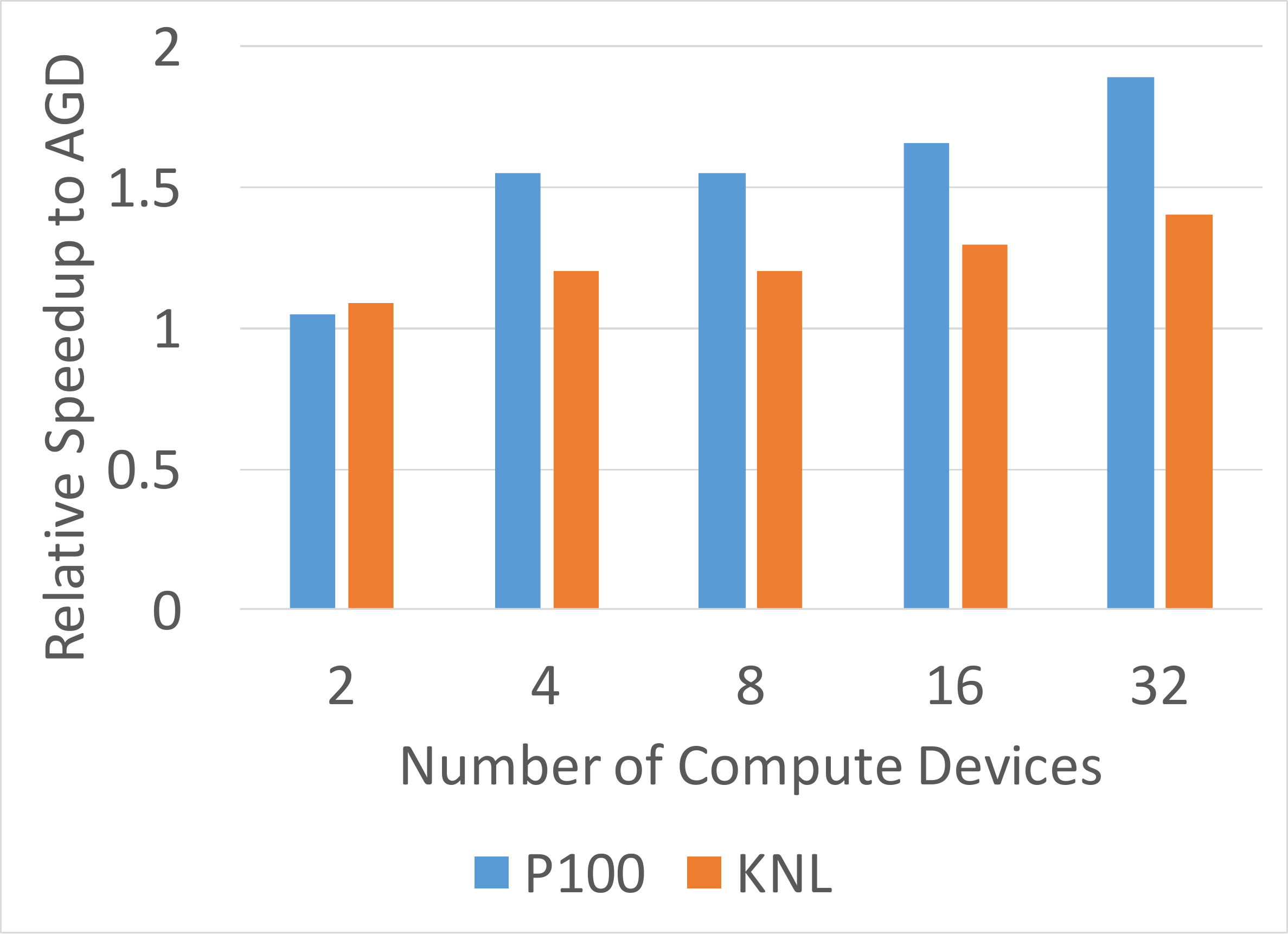}
\caption{\small Relative Speedup of Gossip to AGD for P100 and KNL Clusters on MNIST dataset}
\label{fig:MNIST-speedup}
\end{minipage}
\hfill
\begin{minipage}[t]{0.66\columnwidth}
\centering
\includegraphics[width=\columnwidth]{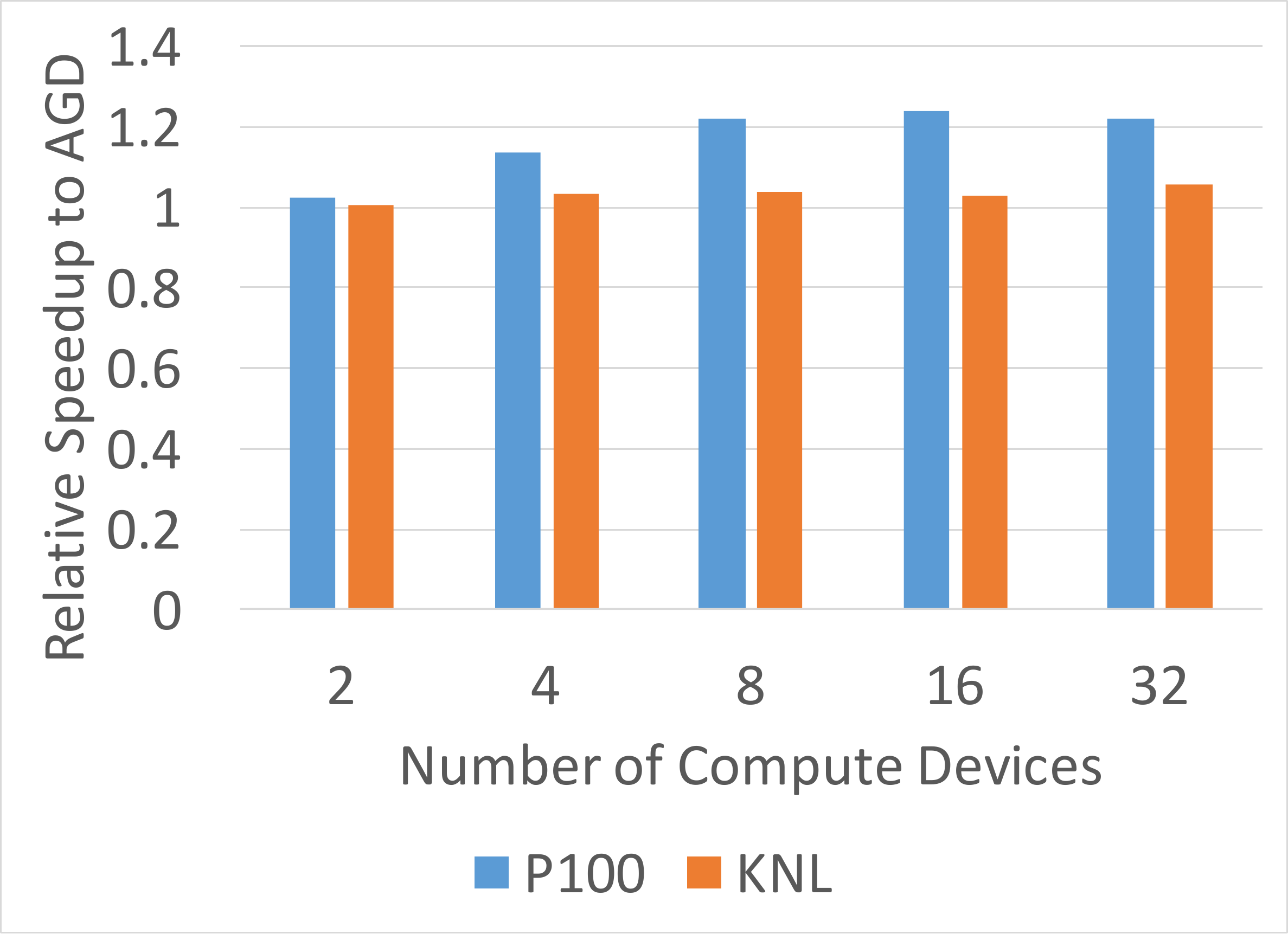}
\caption{\small Relative Speedup of Gossip to AGD for P100 and KNL Clusters on CIFAR10 dataset}
\label{fig:CIFAR10-speedup}
\end{minipage}
\hfill
\begin{minipage}[t]{0.66\columnwidth}
\centering
\includegraphics[width=\columnwidth]{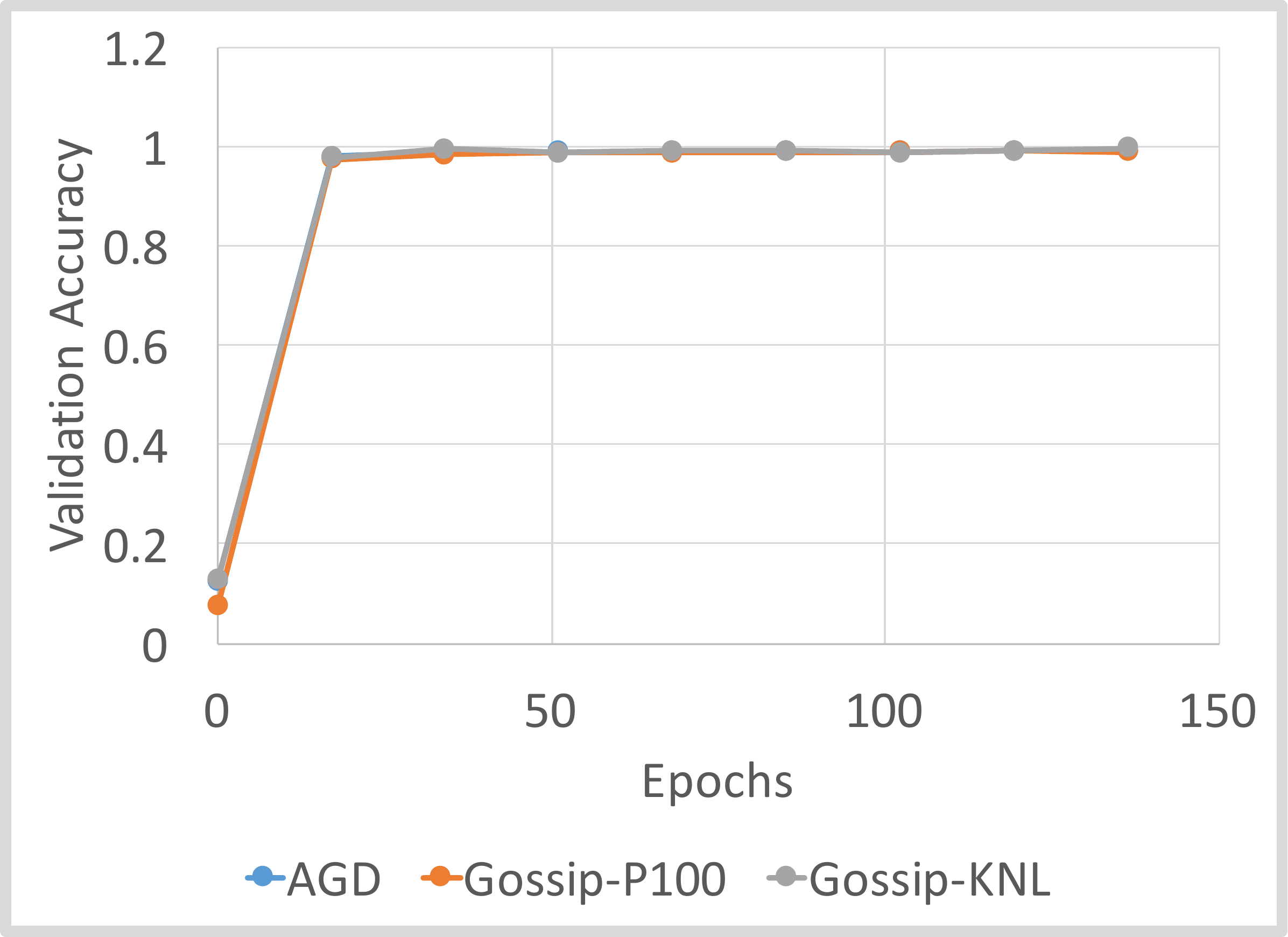}
\caption{\small Validation Accuracies of AGD, Gossip on KNL and Gossip on GPU Clusters on MNIST dataset}
\label{fig:MNIST}
\end{minipage}
\end{figure*}


\subsection{Comparing GossipGraD to AGD: ResNet 50 Evaluation}
\subsubsection{Performance Comparisons}
Table~\ref{table:gossipvspowerai} shows the compute efficiency of GossipGraD and
compares it with published efficiency on PowerAI~\cite{powerai} using up to 128 P100 GPUs.
We use a batch size of 32 on each device for GossipGraD implementations --
which is the batch size used by other researchers. The GossipGraD implementation use
asynchronous \texttt{MPI\_TestAll} based implementation. Since we use
weak-scaling, the overall batch size on 128 GPUs is 4096.  We observe that the
overall compute efficiency of GossipGraD is $\approx$ 100\% independent of the
number of GPUs. The PowerAI implementation achieves 100\% efficiency
for 4 and 8 GPUs but continues to decrease to 95\% on 128 GPUs.
At 32 batch size per device for GossipGraD, the feedforward and
backpropagation time is about 96ms (providing 10.4 batch updates/second).
The synchronous point-to-point communication time is 27ms which is completely
overlapped by the GossipGraD implementation.

\subsubsection{Accuracy Comparisons}
Figure~\ref{fig:RESNET-acc} shows the validation accuracy as a function of
epochs for ResNet50 using GossipGraD on 128 P100 GPUs. 
ResNet50 uses a step learning training regimen, which is executed
for 100 epochs. After every 30 epochs, the current learning rate is multiplied
by 0.1. 
Since ResNet50 is well
studied by other researchers including PowerAI~\cite{powerai}, Caffe2~\cite{goyal:arxiv17} and Chainer~\cite{chainer}, we compare
the GossipGraD validation accuracy after a fixed number of epochs with these
approaches published elsewhere. 

For GossipGraD, we use the original learning rate provided for ResNet50 (0.1)
and follow the training regimen presented earlier. We compute the validation
accuracies after every $\approx$ 5 epochs. At about 30 epochs, we see an accuracy of
50\%. This  which is the accuracy reported by Chainer, PowerAI and Caffe2 for 30
epochs. Hence, GossipGraD provides similar accuracy as well-published
literature while providing nearly perfect overlap of communication with
computation.

\begin{table}[!htbp] 
		\centering
		\begin{tabular}{|c|c|c|c|c|c|c|c|c|}
				\hline
				Name  & 4 & 8 & 16 & 32 & 64 & 128 \\
				\hline 
				GossipGraD  & 100 & 100 & 100 & 100 & 100 & 100 \\
				\hline
				PowerAI  & 100 & 100 & 98 & 99 & 97 & 95 \\
				\hline
		\end{tabular}\\  
		\caption{\small Compute Efficiency (\%) for GossipGraD and PowerAI~\cite{powerai} using up to 128 P100 GPUs. PowerAI performs better than Caffe2, since it uses hierarchical rings based all-to-all reduction}
		\label{table:gossipvspowerai}  
\end{table}


\begin{figure*}[htbp]
\begin{minipage}[t]{0.66\columnwidth}
\centering
\includegraphics[width=\columnwidth]{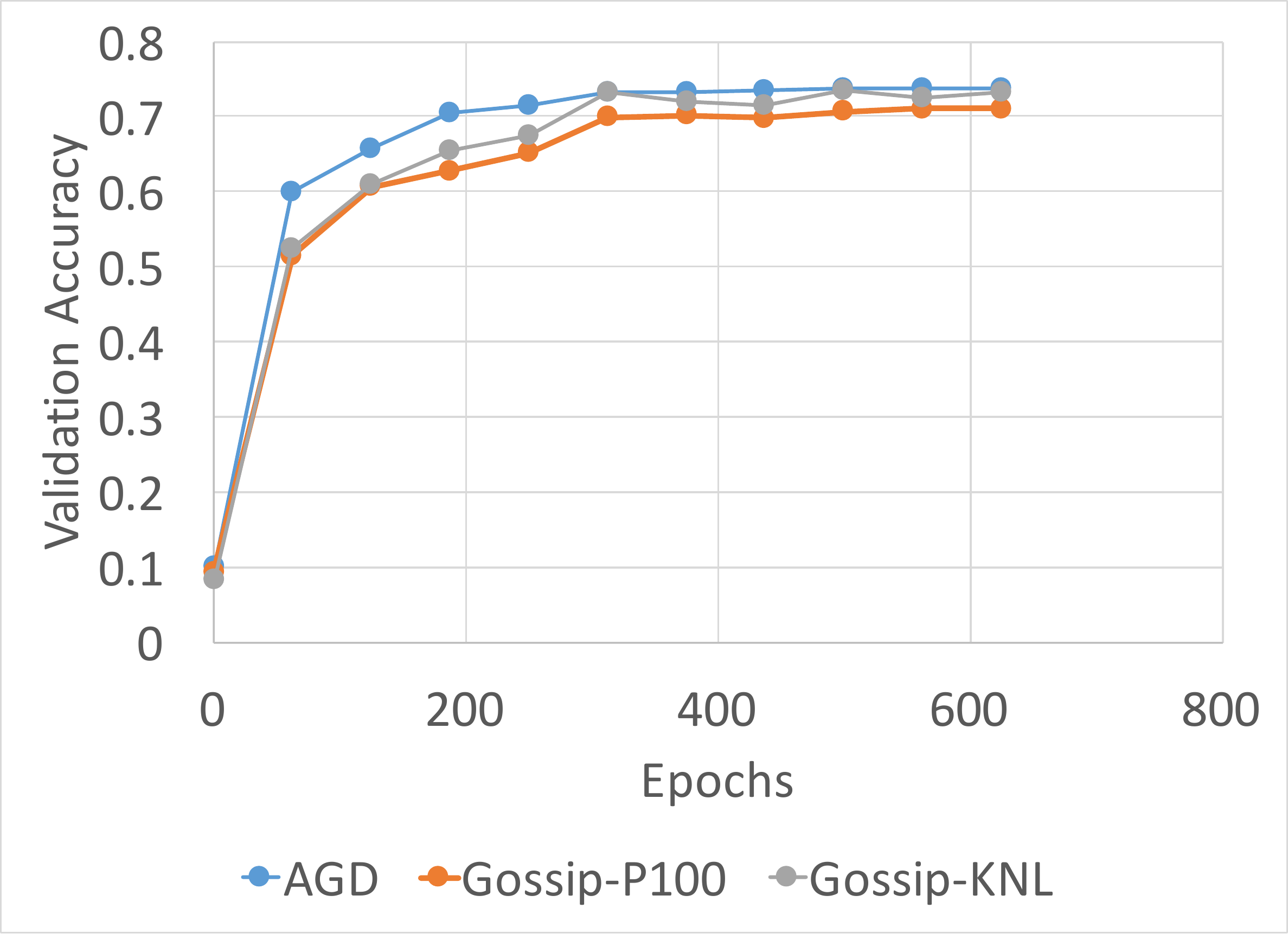}
\caption{\small Validation Accuracies of AGD, GossipGraD on KNL and GossipGraD on GPU Clusters for CIFAR10 dataset}
\label{fig:CIFAR10}
\end{minipage}
\hfill 
\begin{minipage}[t]{0.66\columnwidth}
\centering
\includegraphics[width=\columnwidth]{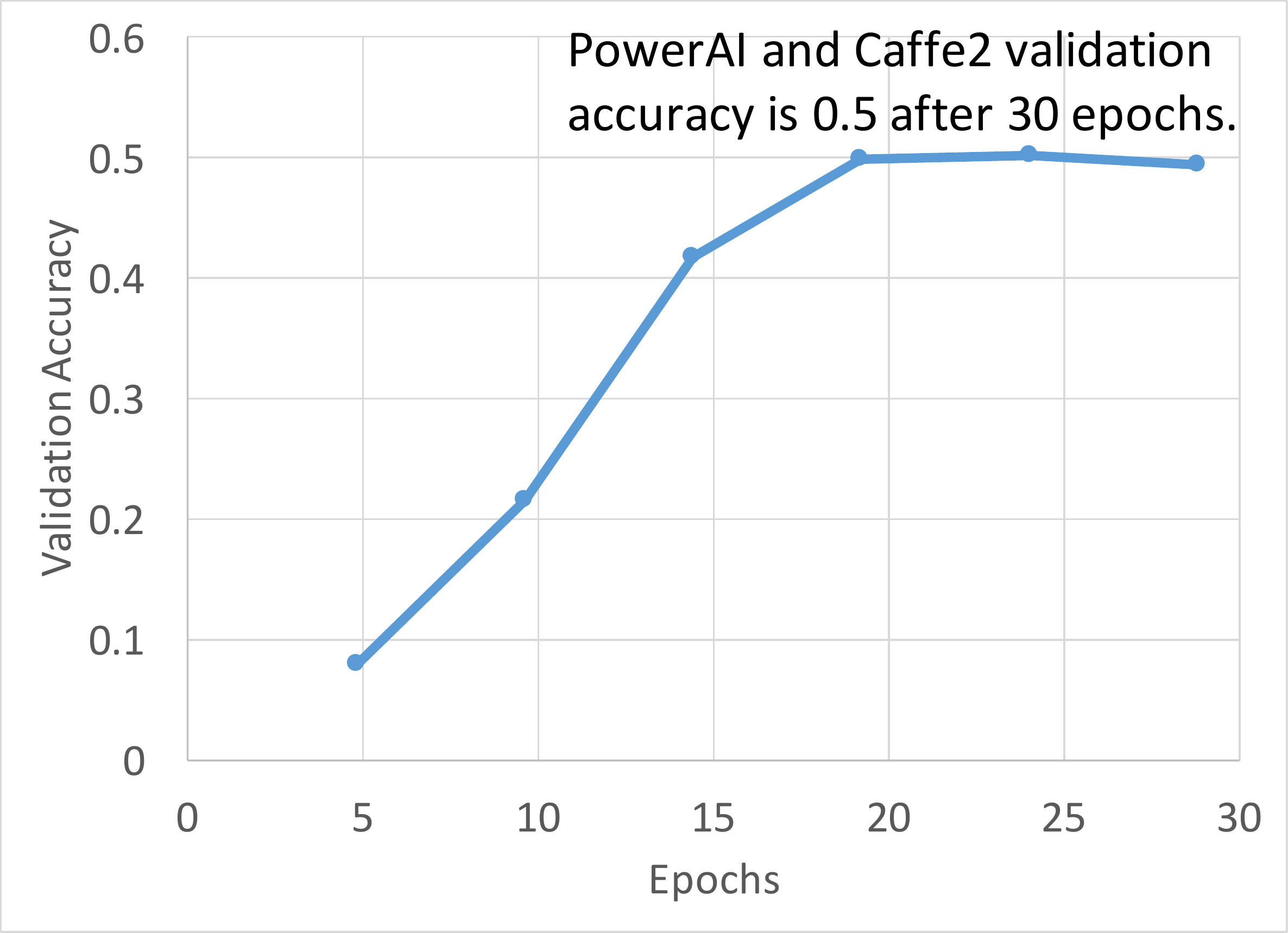}
\caption{\small Accuracy of GossipGraD for 128 P100 GPUs for ResNet50}
\label{fig:RESNET-acc}
\end{minipage}
\hfill
\begin{minipage}[t]{0.66\columnwidth}
\includegraphics[width=\columnwidth]{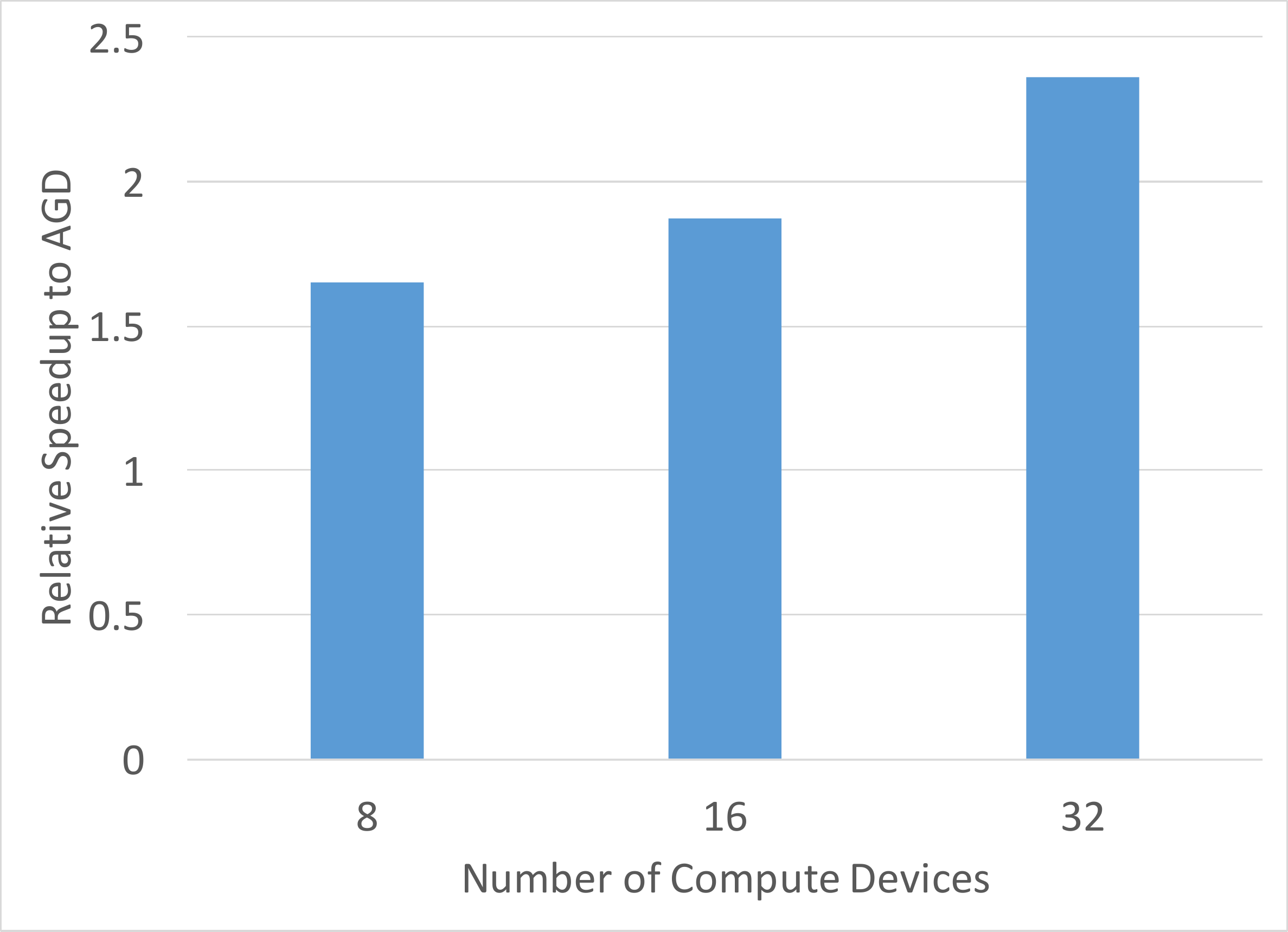}
\caption{\small Relative Speedup of GossipGraD to AGD for P100 Cluster for GoogLeNet}
\label{fig:GNET-speedup}
\end{minipage}
\end{figure*}

\subsection{Comparing GossipGraD with AGD: Evaluating GoogLeNet}
Figure~\ref{fig:GNET-speedup} shows the relative speedup of GossipGraD in
comparison to AGD using up to 32 GPUs for GoogLeNet respectively. We use a
batch size of 16 on each compute node for both GossipGraD and AGD. We observe:
1) relative speedup for GossipGraD increases with scale -- even on weak scaling
-- since the overall communication time increases. GoogLeNet has 5M parameters
(20M floats) -- which is much smaller than ResNet which has 25M parameters.
However, ResNet50 is relatively computationally expensive than GoogLeNet.

It is worthwhile observing that for GoogLeNet as well, we observe a $\approx$
100\% computation efficiency at all compute nodes -- which implies that
GossipGraD is effective for several architectures, and independent of the neural network topology.
Figure~\ref{fig:GNET-loss} shows the loss chart for GoogLeNet by comparing AGD
and GossipGraD using 32 P100 GPUs. We compare the loss over time. During this
time, GossipGraD has only covered 10\% of overall iterations -- which implies
that there is a significant time remaining in the overall training phase.
However, even at this short time GossipGraD provides similar or better loss
(since lower loss is better) in comparison to AGD.

\begin{figure}[htbp]
\includegraphics[width=\columnwidth]{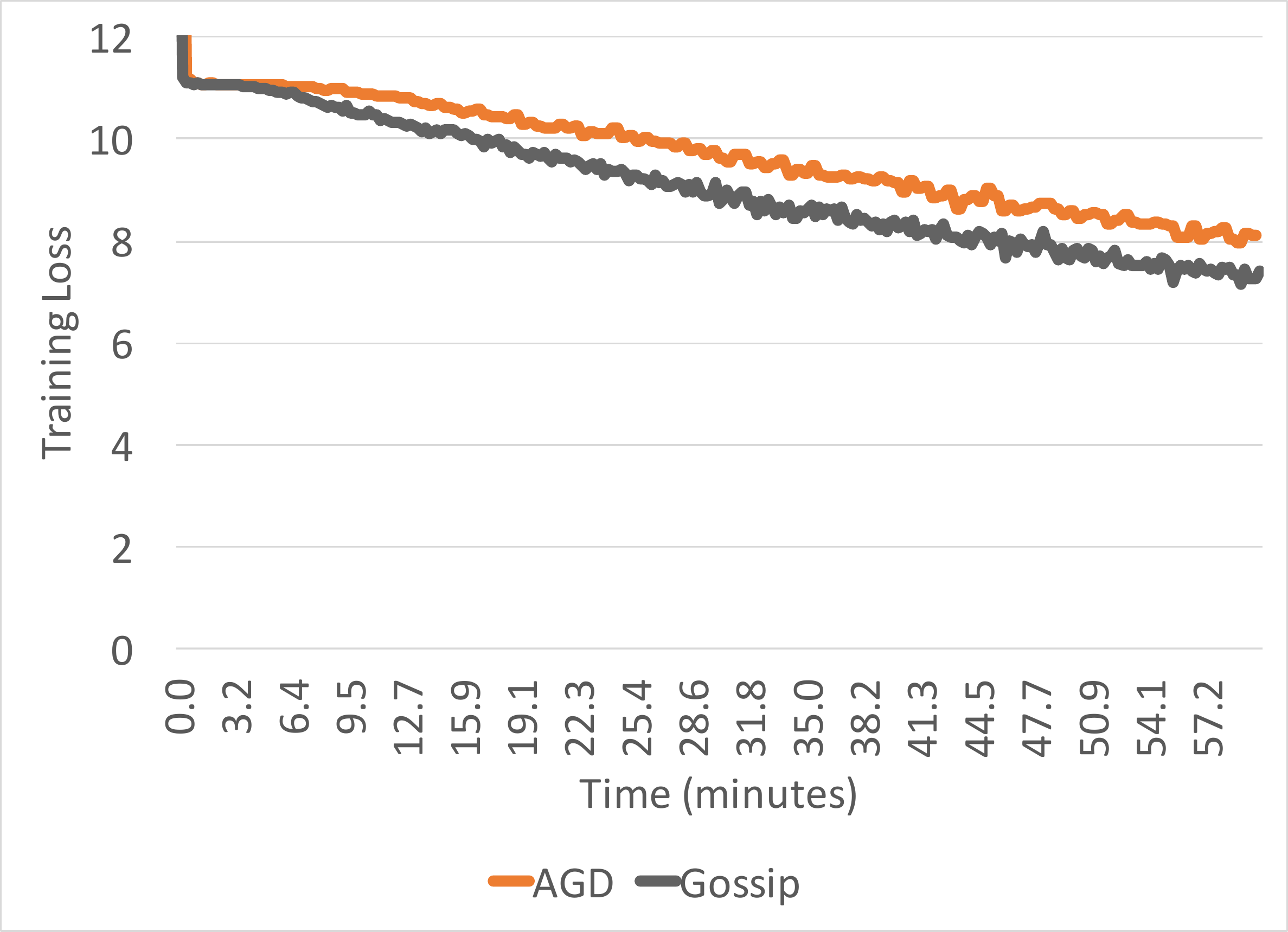}
\caption{\small Training Loss after one hour of GossipGraD and AGD on 32 P100 GPUs with GoogLeNet}
\label{fig:GNET-loss}
\end{figure}

\subsection{Comparing GossipGraD with AGD Communicating Every Log(p) Steps}
Another approach to achieve $O(1)$ communication is to combine the
models every $log(p)$ steps instead of every step.
Figure~\ref{fig:mnist-logp} compares the performance of this approach to
GossipGraD using the LeNet3 network. If the cost of the $log(p)$
reduction (AGD) can be amortized over $log(p)$ steps the performance is
improved. For the every-$log(p)$ approach, the performance is trending
slightly positive while GossipGraD remains flat. However, the
performance of GossipGraD is still greater.  Though these two approaches
might eventually perform similarly at large scales, the effect on
validation accuracy cannot be ignored. Though we expect all approaches
to reach target accuracy if tuned with appropriate hyperparameters, for
those cases in Figure~\ref{fig:mnist-logp} only GossipGraD was learning.
This further shows that GossipGraD is less susceptible to incorrect
hyperparameter tuning as one scales.

\begin{figure}[htbp]
\includegraphics[width=\columnwidth]{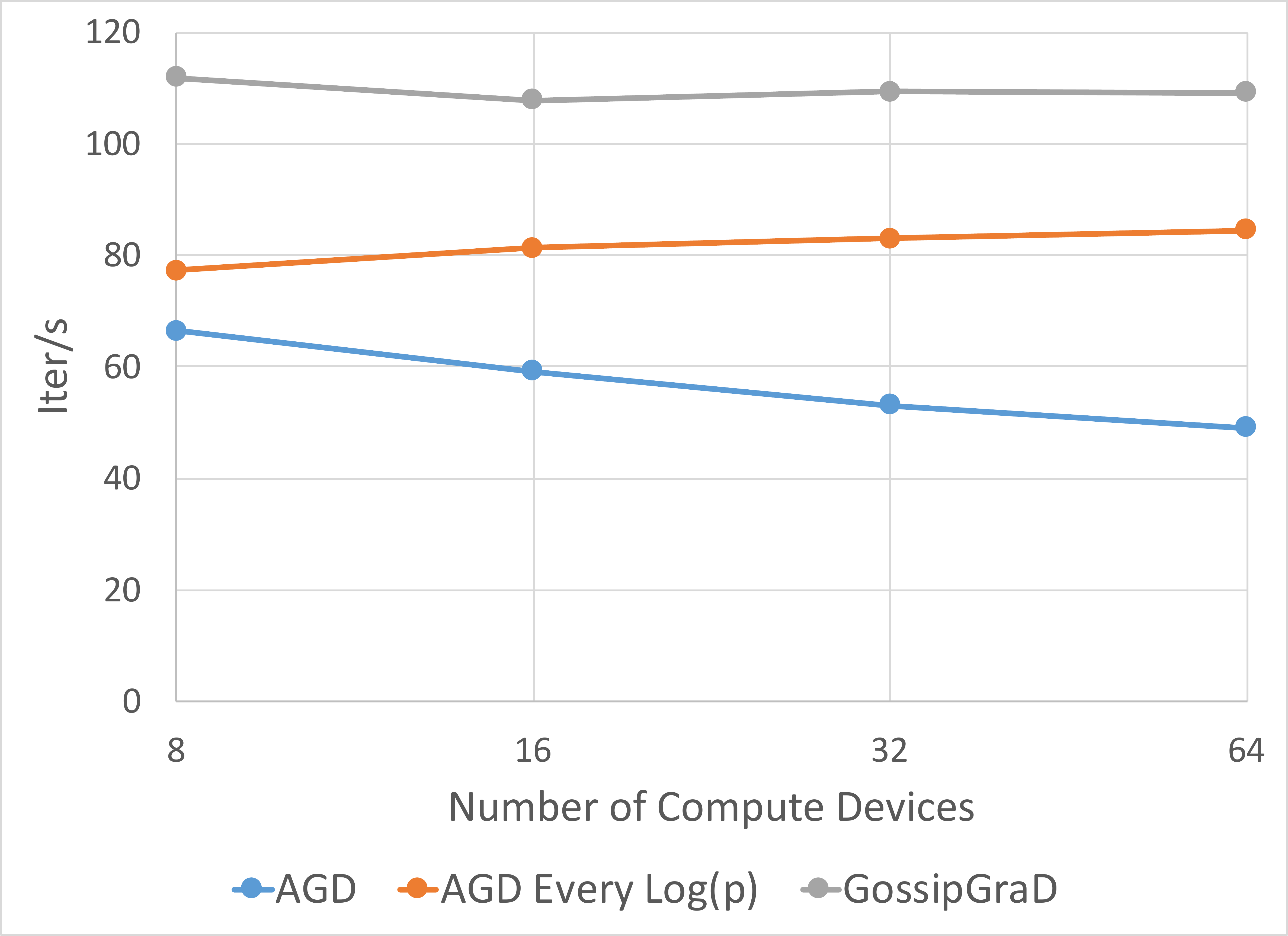}
\caption{\small Performance of GossipGrad versus Computing AGD Every Log(p) Iterations}
\label{fig:mnist-logp}
\end{figure}

\subsection{Discussion}
We evaluated GossipGraD on several dimensions including different datasets,
their associated neural network topologies, and GPU/CPU architectures on the
metrics of performance and accuracy.  We observe that GossipGraD provides
complete overlap of communication with computation (as expected), and similar
accuracy/validation loss as well studied SGD/AGD. GossipGraD also does not
require significant hyper-parameter tuning -- which makes it an attractive
option for generic datasets. A possible hyper-parameter to tune would be to have a batch size just small enough which can be completely overlapped with computation -- however, we have avoided making any changes to batch sizes so that our results are easily reproducible using hyper-parameters available on NVIDIA-Caffe and Intel-Caffe.

\section{Related Work}
\label{sec:related}
Gradient Descent (GD) is the most widely used algorithm for implementing Deep
Learning (DL) algorithms.  Since GD is slow, a variant which selects a random
subset ({\em batch}) of the original dataset is used -- referred as
batch/stochastic gradient descent (SGD).  Several frameworks provide high
performance implementations of SGD.  The most widely used implementations are
Caffe~\citep{jia2014caffe} (CPUs/GPUs), Warp-CTC (GPUs),
Theano~\citep{Bastien-Theano-2012, bergstra+al:2010-scipy} (CPUs/GPUs),
Torch~\citep{Collobert02torch:a} (CPUs/GPUs), CNTK~\citep{export:226641} (GPUs
and Distributed Memory using MPI) and Google
TensorFlow~\citep{tensorflow2015-whitepaper} which use NVIDIA CUDA Deep Neural
Network (cuDNN) library. 
We use SGD as the baseline for designing and
implementing GossipGraD and its variants.  As evident, deep neural networks
(DNN) -- which store the model for DL algorithms -- suffer from a variety of
problems, such as vanishing gradient~\citep{Bianchini2014}.  Hinton {\em et
al.}~\citep{Hinton06afast, NIPS2006_3048} proposed a solution to address this
problem, by layer-wise training
\textit{autoencoders}~\citep{HintonSalakhutdinov2006b}. Individual layers may
be coalesced together to form a single
DNN~\citep{Vincent:2010:SDA:1756006.1953039}. 

Several researchers have proposed methods to scale DL algorithms on distributed
memory systems. Table~\ref{table:newcomparisons} shows a table of distributed DL implementations with comparisons on several metrics.
We classify these approaches among {\em parameter server} based
and non-parameter server-based approaches.
Distbelief~\citep{NIPS2012_0598} is an approach proposed by Dean {\em et al.},
which uses a parameter server (PS) for model updates at a central server.
Several PS optimizations specific to GPU have been proposed in approaches such
as GeePS~\cite{geeps} and MXNET~\cite{mxnet}, and techniques to address delay
compensation~\cite{zheng:arxiv16}. Chen {\em et al.} have studied the
limitations of PS and reported that parameter servers quickly become a
bottleneck, require expensive {\em warm up phase} (during which training is
conducted on a single compute node, till loss starts to decrease). 
However, existing
Spark implementations typically use sockets -- which is not optimal for HPC
interconnects.  To leverage HPC interconnects effectively, several Remote
Direct Memory Access (RDMA) based PS approaches have been
proposed~\cite{deepimage, li:nips14}.  However, they primarily suffer from
typical convergence problems and performance bottlenecks as observed by other
researchers~\cite{li:nips14}.

Several researchers have proposed techniques to address the limitations of
parameter server based DL implementations~\cite{de:arxiv15, scaffe,matex,
zheng:arxiv16, das:arxiv16}.  These DL implementations use model/data
parallelism. Das {\em et al.} have proposed hybrid parallelism for scaling DL
implementations. However, given the increasing depth of convolutional layers
and suggested by Krizhevsky~\cite{oneweirdtrick}, the majority of DL
implementations primarily focus on data parallelism -- which is the focus of
this paper as well. 

Data parallelism approaches for scaling DL require hyperparameter tuning
as the scale increases~\cite{oneweirdtrick,you2017imagenet}. Though this
work did not focus on hyperparemter tuning, we expect GossipGraD would
also benefit from hyperparameter tuning such as the proposed LARS
method~\cite{you2017imagenet,you2017scaling}, RMSprop
warm-up~\cite{tieleman2012}, and slow-start learning rate
schedules~\cite{akiba2017,goyal2017}.

\section{Conclusions}
\label{sec:conclusions}
In this paper, we have presented GossipGraD, which is a novel gossip communication protocol
based Stochastic Gradient Descent (SGD) algorithm for scaling Deep Learning
(DL) algorithms on large-scale systems. GossipGrad has reduced the overall communication complexity from $\Theta(\log(p))$ for
$p$ compute nodes in well-studied SGD to $O(1)$ and considered diffusion such that
compute nodes exchange their updates (gradients) indirectly after every
$\log(p)$ steps. It also considers rotation of communication partners for facilitating direct
diffusion of gradients and  asynchronous distributed shuffle of samples during
the feedforward phase in SGD to prevent over-fitting.
We have implemented GossipGraD for GPU and CPU clusters and use
NVIDIA GPUs (Pascal P100) connected with InfiniBand, and Intel Knights Landing
(KNL) connected with Aries network. We evaluate GossipGraD using well-studied
dataset ImageNet-1K ($\approx$ 250GB), and widely studied neural network
topologies such as GoogLeNet and ResNet50.
Our performance evaluation
using both KNL and Pascal GPUs has indicated  that GossipGraD can achieve perfect
efficiency for these datasets and their associated neural network topologies.
Specifically, for ResNet50, GossipGraD is able to achieve $\approx$ 100\% compute
efficiency using 128 NVIDIA Pascal P100 GPUs -- while matching the top-1 classification accuracy
published in literature.

\bibliographystyle{IEEEtran}
\balance
\bibliography{distdl,FasterLearning,extra,vishnu,apoptosis,agd,mathstat}

\newpage
\appendix
\section{Artifact Description: [GossipGraD: Gossip Protocol Based Asynchronous Gradient Descent for Scalable Deep Learning]}

\subsection{Abstract}

This artifact description contains information needed to reproduce the
software environments used by the experiments of the associated
submission. The source code contains changes to open source
software packages Intel-Caffe and NVIDIA-Caffe.

\subsection{Description}

\subsubsection{Check-list (artifact meta information)}

{\small
\begin{itemize}
    \item {\bf Algorithm: } Deep Learning Networks AlexNet and GoogLeNet
\item {\bf Program: } Intel-Caffe, NVIDIA-Caffe
  \item {\bf Compilation: } Sandy-Bridge -- intel/16.1.150; Pascal -- IBM XL C/C++ for Linux, V13.1.5
  \item {\bf Data set: } ImageNet ilsvrc2012
  \item {\bf Hardware: } 2 10-core Sandy-Bridge sockets and 128GB memory per
      node; 2 10-core IBM POWER8 CPUs 256GB RAM with 4 NVIDIA Tesla P100
      GPUs 16GB HBM2 memory per node
  \item {\bf Experiment workflow: } Clone architecture-specific
      versions of Caffe, install all caffe dependencies, apply custom
      patches to add parallel netcdf reader; parallel algorithms; and
      command-line arguments for caffe tool, compile, run caffe tool
      with updated solver parameters.
  \item {\bf Experiment customization: } Upstream-compatible patches
      were created for Intel-Caffe and NVIDIA-Caffe that add a parallel
      netcdf DataLayer as well as additional parallel modules. The
      parallel modules wrap a standard Caffe::Solver instance and
      implement the Caffe::Solver::Callback interface in order to
      communicate network data and biases as needed. Lastly, mini-batch
      sizes were set to 16 instead of their customary defaults for
      AlexNet and GoogLeNet.
  \item {\bf Publicly available?: } Yes
\end{itemize}
}

\subsubsection{How software can be obtained (if available)}

Intel-Caffe is available from https://github.com/intel/caffe. Our
software branches from the master branch at commit ID
dated 29 December 2016\\
0bc848bd32d17d5f17bfc7e20915e9805fd1f180 and is available at
https://github.com/matex-org/caffe-intel.

NVIDIA-Caffe is available from https://github.com/NVIDIA/caffe. Our
software branches from the master branch at commit ID
dated 27 December 2016 6d723362f0f7fe1aaba7913ebe51cc59b12c0634 and is
available at https://github.com/matex-org/caffe-nvidia.

\subsubsection{Hardware dependencies}

Intel-Caffe is optimized for Intel CPUs, specifically those featuring
advanced vector intrinsics such as AVX. Intel-Caffe automatically
downloads a custom MKL library for use with Intel-Caffe if a sufficient
version is not automatically located.

NVIDIA-Caffe is specifically tuned for GPUs feauting the latest CUDA and
cuDNN versions, specifcally CUDA 8.0 and cuDNN 5.1 at the time of
manuscript preparation. The Pascal cluster evaluation took
advantage of coherent memory between the GPU and CPU for MPI
communication using IBM Spectrum MPI, however this feature is not
strictly required for the algorithms and advances presented.

\subsubsection{Software dependencies}

Both Intel-Caffe and NVIDIA-Caffe share the same list of software
dependencies with their upstream BVLC-Caffe implementation. IBM also
provides a version of Caffe that will compile on the POWER architecture
but is otherwise identical to BVLC-Caffe. We copied the specific POWER
assembly code for `pause' that was necessary to compile on the Pascal
cluster from their github repository at https://github.com/ibmsoe/caffe.

\subsubsection{Datasets}

To evaluate AlexNet and GoogLeNet networks we use
ImageNet Large Scale Visual Recognition Challenge 2012 (ILSVRC2012). For
downloading ImageNet database please follow url:
http://caffe.berkeleyvision.org/gathered/examples/imagenet.html.
Standard Caffe installations have instructions for creating LMDB
databases from these datasets.

We customize the datasets by converting the LMDB databases into netCDF
files using a Python script and the pupynere Python package for creating
netCDF files. The images were stored as a byte datatype with a
corresponding integer for the classification labels. Parallel netCDF was
then used for parallel reading of the datasets.

\subsection{Installation}

Intel-Caffe and NVIDIA-Caffe, with or without our modifications, install
as if following the well-known instrutions for the upstream BVLC-Caffe
with only a few minor exceptions. An MPI compiler is used by
setting the CUSTOM\_CXX variable in the Makefile.config to use, e.g.,
mpicxx. The parallel netCDF library is not customarily used with Caffe
and must be added to the linker flags. Intel-Caffe recognizes and uses
the preprocessor symbol USE\_MPI and it must be set to 1. 

\subsection{Experiment workflow}

Caffe uses two related prototxt files per evaluation, one to describe
the solver and the other to describe the network. We used prototxt files
that come standard with any Caffe distribution. We modified the network
prototxt files to use our custom parallel netCDF reader. Otherwise, we
modified the network prototxt files to use a batch size of 16 for the
majority of our experiments.

The majority of our experiments were evaluating weak scaling and thus
used the same prototxt input files.  The only exception was when weak
scaling the AGD runs.  We increased the learning rate each time we
doubled the number of compute devices, e.g., Sandy-Bridge nodes, GPUs, 
based on the observations made by Krizhevsky et al., increasing
by a factor of $\sqrt 2$ each time.

\subsection{Evaluation and expected result}

Intel-Caffe and NVIDIA-Caffe output is unchanged with respect to
BVLC-Caffe. Any traditional analysis of such output remains valid. Any
reported accuracy or loss numbers come directly from this standard Caffe
output.

\subsection{Experiment customization}

Besides the changes made to customary batch sizes, the experimentation
could be considered as completely customized due to our use of
proprietary communication code additions to available open source
software packages.


\end{document}